\documentclass[english]{article}
\usepackage[T1]{fontenc}
\usepackage[latin9]{inputenc}
\usepackage{bm}
\usepackage{amsmath}
\usepackage{amsthm}
\usepackage{amssymb}
\usepackage{todonotes}
\usepackage{booktabs}
\usepackage{placeins} %for FloatBarrier
\usepackage{hyperref}
\usepackage{cleveref}
\usepackage{lipsum}  
\usepackage{float}
\usepackage{enumerate}
\usepackage{enumitem}
\usepackage{units}

\usepackage{caption} 
\makeatletter
\theoremstyle{plain}
\newtheorem{lem}{\protect\lemmaname}
\theoremstyle{remark}
\newtheorem{rem}{\protect\remarkname}
\theoremstyle{plain}
\newtheorem{thm}{\protect\theoremname}
\DeclareMathOperator*{\plim}{plim}

\makeatother

\usepackage{babel}
\providecommand{\lemmaname}{Lemma}
\providecommand{\remarkname}{Remark}
\providecommand{\theoremname}{Theorem}

\title{
    Shapley value confidence intervals for\\
    attributing variance explained
}

\author{
    Daniel Vidali Fryer\footnote{\texttt{
    daniel.fryer@uq.edu.au} (Corresponding Author)},\\
    School of Mathematics and Physics,\\
    The University of Queensland, St Lucia, Australia \\
    Inga Str{\" u}mke\footnote{\texttt{inga@simula.no}},\\
    Simula Research Laboratory,\\
    Pilestredet 52, Oslo, Norway\\
    Hien Nguyen\footnote{\texttt{h.nguyen5@latrobe.edu.au}}, \\
    Department of Mathematics and Statistics, \\
    La Trobe University, Melbourne, Australia\\
    }
\date{}

\begin{document}

\maketitle
 % ---------------------------------------------------------

\begin{abstract}
    The coefficient of determination, the $R^2$, is often used to measure the variance explained by an affine combination of multiple explanatory covariates. 
    An attribution of this explanatory contribution to each of the individual covariates is often sought in order to draw inference regarding the importance of each covariate with respect to the response phenomenon.
%    A ranking of the explanatory contribution of each of the individual covariates is often sought in order to draw inference regarding the importance of each covariate with respect to the response phenomenon. 
    A recent method for ascertaining such an attribution is via the game theoretic Shapley value decomposition of the coefficient of determination. 
    Such a decomposition has the desirable efficiency, monotonicity, and equal treatment properties.  
    Under a weak assumption that the joint distribution is pseudo-elliptical, we obtain the asymptotic normality of the Shapley values. 
    We then utilize this result in order to construct confidence intervals and hypothesis tests regarding such quantities. 
    Monte Carlo studies regarding our results are provided. 
    We found that our asymptotic confidence intervals are computationally superior to competing bootstrap methods and are able to improve upon the performance of such intervals. 
    In an expository application to Australian real estate price modelling, we employ Shapley value confidence intervals to identify significant differences between the explanatory contributions of covariates, between models, which otherwise share approximately the same $R^2$ value. 
    These different models are based on real estate data from the same periods in 2019 and 2020, the latter covering the early stages of the arrival of the novel coronavirus, COVID-19.
\end{abstract}

% ---------------------------------------------------------

% ---------------------------------------------------------
\section{Introduction}
% ---------------------------------------------------------
The multiple linear regression model (MLR) is among the most commonly
applied tools for statistics inference; see \cite{Seber:2003aa} and
\cite[Part I]{Greene:2007aa} for thorough introductions to MLR models.
In the usual MLR setting, one observes an independent and identically distributed
(IID) sample of data pairs $\bm{Z}_{i}^{\top}=\left(Y_{i},\bm{X}_{i}^{\top}\right)\in\mathbb{R}^{d+1}$,
where $i\in\left[n\right]=\left\{ 1,\dots,n\right\} $, and $d,n\in\mathbb{N}$.
The MLR model is then defined via the linear relationship
\[
\mathbb{E}\left(Y_{i}|\bm{X}_{i}=\bm{x}_{i}\right)=\beta_{0}+\sum_{i=1}^{d}\beta_{j}X_{ij}\text{,}
\]
where $\bm{\beta}^\top=\left(\beta_{0},\dots,\beta_{n}\right)\in\mathbb{R}^{d+1}$,
and $\bm{X}_{i}^{\top}=\left(X_{i1},\dots,X_{id}\right)$. We shall
also write $\bm{Z}_{i}^{\top}=\left(Z_{i0},Z_{i1},\dots,Z_{id}\right)$,
when it is convenient to do so. Here, $\left(\cdot\right)^\top$ is the transposition operator.

The usual nomenclature is to call the $Y_{i}$ and $\bm{X}_{i}$ elements
of each pair, the response (or dependent) variable and the explanatory
(or covariate) vector, respectively. Here, the $j$th element of $\bm{X}_{i}$: 
$X_{ij}$ ($j\in\left[d\right]$), is referred to as the $j\text{th}$ explanatory
variable (or the $j\text{th}$ covariate). We may put the pairs of
data into a dataset $\mathcal{Z}_{n}=\left\{ \bm{Z}_{i}\right\} _{i=1}^{n}$.

Let $R_{jk}\left(\mathcal{Z}_{n}\right)$ denote the sample correlation
coefficient
\begin{equation}
R_{jk}\left(\mathcal{Z}_{n}\right)=\frac{\sum_{i=1}^{n}\left(Z_{ij}-\bar{Z}_{j}\right)\left(Z_{ik}-\bar{Z}_{k}\right)}{\sqrt{\sum_{i=1}^{n}\left(Z_{ij}-\bar{Z}_{j}\right)^{2}}\sqrt{\sum_{i=1}^{n}\left(Z_{ik}-\bar{Z}_{k}\right)^{2}}}\text{,}\label{eq: sample corr}
\end{equation}
for each $j,k\in\left\{ 0\right\} \cup\left[d\right]$. Here $\bar{Z}_{j}=n^{-1}\sum_{i=1}^{n}Z_{ij}$
is the sample mean of variable $j$. Write $\mathcal{U}\subseteq\left\{ 0\right\} \cup\left[d\right]$
be a nonempty subset, where $\mathcal{U}=\left\{ u_{1},\dots,u_{\left|\mathcal{U}\right|}\right\} $,
where $\left|\mathcal{U}\right|$ is the cardinality of $\mathcal{U}$.
We refer to the matrix of correlations between the variables in $\mathcal{U}$
as
\begin{equation}
\mathbf{C}_{n}\left(\mathcal{U}\right)=
\left[\begin{array}{cccc}
1 & R_{u_{1}u_{2}}\!\left(\mathcal{Z}_{n}\right) & \cdots & R_{u_{1}u_{\left|U\right|}}\!\left(\mathcal{Z}_{n}\right)\\
R_{u_{2}u_{1}}\!\left(\mathcal{Z}_{n}\right) & \phantom{\Big|} 1 & \cdots & R_{u_{2}u_{\left|\mathcal{U}\right|}}\!\left(\mathcal{Z}_{n}\right)\\
\vdots & \vdots & \ddots & \vdots\\
R_{u_{\left|\mathcal{U}\right|}u_{1}}\!\left(\mathcal{Z}_{n}\right) & R_{u_{\left|\mathcal{U}\right|}u_{2}}\!\left(\mathcal{Z}_{n}\right) & \cdots & 1
\end{array}\right]
\text{.}\label{eq: correlation S}
\end{equation}
A common inferential task is to determine the degree to which the
response can be explained by the covariate vector, in totality. The
usual device for addressing this question is via the coefficient of
determination (or squared coefficient of multiple correlation), which
is defined as

\begin{equation}
R^{2}\left(\mathcal{Z}_{n}\right)=1-\frac{\left|\mathbf{C}_{n}\left(\left\{ 0\right\} \cup\left[d\right]\right)\right|}{\left|\mathbf{C}_{n}\left(\left[d\right]\right)\right|}\text{,}\label{eq: R2 definition}
\end{equation}
where $\left|\mathbf{C}\right|$ is the matrix determinant of $\mathbf{C}$
(cf. \cite{Cowden:1952aa}). Intuitively, the coefficient of determination
measures the proportion of the total variation in the response variable
that is explained by variation in the covariate vector. See \cite[Sec. 4.4]{Seber:2003aa}
and \cite[Sec. 3.5]{Greene:2007aa} for details regarding the derivation
and interpretation of the $R^{2}\left(\mathcal{Z}_{n}\right)$ coefficients.

A refinement to the question that is addressed by the $R^{2}\left(\mathcal{Z}_{n}\right)$
coefficient, is that of eliciting the contribution of each of the covariates to the total value of $R^{2}\left(\mathcal{Z}_{n}\right)$.

In the past, this question has partially been resolved via the use
of partial correlation coefficients (see, e.g. \cite[Sec. 3.4]{Greene:2007aa}).
Unfortunately, such coefficients are only able to measure
the contribution of each covariate coefficient of determination, conditional
to the presence of other covariates that are already in the MLR model.

A satisfactory resolution to the question above, is provided by \cite{Lipovetsky:2001aa},
\cite{Israeli:2007aa}, and \cite{Huettner:2012aa}, who each suggested
and argued for the use of the Shapley decomposition of \cite{Shapley:1953aa}.
The Shapley decomposition is a game-theoretic method for decomposing
the contribution to the value of a utility function in the context
of cooperative games.

Let $\bm{\pi}^{\top}=\left(\pi_{1},\dots,\pi_{d}\right)$ be a permutation
of the set $\left[d\right]$. For each $j\in\left[d\right]$, let
\[
\mathcal{S}_{j}\left(\bm{\pi}\right)=\left\{ k \,:\, \pi_k < \pi_j, \, k\in\left[d\right]\right\} 
\]
be the elements of $[d]$ that appear before $\pi_j$ when $[d]$ is permuted by $\bm{\pi$}. We may
define $R_{\mathcal{S}_{j}\left(\bm{\pi}\right)}^{2}\left(\mathcal{Z}_{n}\right)$
and $R_{\left\{ j\right\} \cup\mathcal{S}_{j}\left(\bm{\pi}\right)}^{2}\left(\mathcal{Z}_{n}\right)$
in a similar manner to \eqref{eq: R2 definition}, using the generic
definition
\begin{equation}
R_{\mathcal{S}}^{2}\left(\mathcal{Z}_{n}\right)=1-\frac{\left|\mathbf{C}_{n}\left(\left\{ 0\right\} \cup\mathcal{S}\right)\right|}{\left|\mathbf{C}_{n}\left(\mathcal{S}\right)\right|}\text{,}\label{eq: R2 S}
\end{equation}
for nonempty subsets $\mathcal{S}\subseteq\left[d\right]$, and  $R_{\mathcal{\{\,\}}}^{2}\left(\mathcal{Z}_{n}\right)=0$ for the empty set.

Treating the coefficient of determination as a utility function, we may conduct a Shapley partition of the $R^{2}\left(\mathcal{Z}_{n}\right)$ coefficient by computing the $j\text{th}$ Shapley value, for each of the $d$ covariates, defined by
\begin{equation}
V_{j}\left(\mathcal{Z}_{n}\right)=\left|\mathcal{P}\right|^{-1}\sum_{\bm{\pi}\in\mathcal{P}}\left[R_{\left\{ j\right\} \cup\mathcal{S}_{j}\left(\bm{\pi}\right)}^{2}\left(\mathcal{Z}_{n}\right)-R_{\mathcal{S}_{j}\left(\bm{\pi}\right)}^{2}\left(\mathcal{Z}_{n}\right)\right]\text{,}\label{eq: shapley value}
\end{equation}
where $\mathcal{P}$ is the set of all possible permutations of $\left[d\right]$.

Compared to other decompositions of the coefficient of determination, such as those considered in \cite{Gromping:2006aa} and \cite{Gromping:2007aa}, the Shapley values, obtained from the partitioning above, have the favorable axiomatic properties that were well exposed in \cite{Huettner:2012aa}.
Specifically, the Shapley values have the efficiency, monotonicity, and, equal treatment properties, and the decomposition is provably the only method that satisfies all three of these properties (cf. \cite[Thm. 2]{Young:1985aa}).
Here, in the context of the coefficient of determination, efficiency, monotonicity, and equal treatment are defined as follows:

\begin{description}

\item[Efficiency:] The sum of the Shapley values across all
covariates equates to the coefficient of determination, that is
\[
\sum_{j=1}^{d}V_{j}\left(\mathcal{Z}_{n}\right)=R^{2}\left(\mathcal{Z}_{n}\right).
\]

\item[Monotonicity:] For pairs of samples $\mathcal{Z}_{m}$ and $\mathcal{Z}_{n}$, of sizes $m,n\in\mathbb{N}$, 
\[
R_{\left\{ j\right\} \cup\mathcal{S}_{j}\left(\bm{\pi}\right)}^{2}\left(\mathcal{Z}_{n}\right)-R_{\mathcal{S}_{j}\left(\bm{\pi}\right)}^{2}\left(\mathcal{Z}_{n}\right)\ge R_{\left\{ j\right\} \cup\mathcal{S}_{j}\left(\bm{\pi}\right)}^{2}\left(\mathcal{Z}_{m}\right)-R_{\mathcal{S}_{j}\left(\bm{\pi}\right)}^{2}\left(\mathcal{Z}_{m}\right)\text{,}
\]
for every $\bm{\pi}\in\mathcal{P}$, implies that $V_{j}\left(\mathcal{Z}_{n}\right)\ge V_{j}\left(\mathcal{Z}_{m}\right)$, for each $j\in[d]$.

\item[Equal treatment:] If covariates
$j,k\in\left[d\right]$ are substitutes in the sense that
\[
R_{\left\{ j\right\} \cup\mathcal{S}_{j}\left(\bm{\pi}\right)}^{2}\left(\mathcal{Z}_{n}\right)=R_{\left\{ k\right\} \cup\mathcal{S}_{k}\left(\bm{\pi}\right)}^{2}\left(\mathcal{Z}_{n}\right)\text{,}
\]
for each $\bm{\pi}\in\mathcal{P}$ such that $k\notin\mathcal{S}_{j}\left(\bm{\pi}\right)$
and $j\notin\mathcal{S}_{k}\left(\bm{\pi}\right)$, then the Shapley decomposition is such that $V_{j}\left(\mathcal{Z}_{n}\right)=V_{k}\left(\mathcal{Z}_{n}\right)$.

\end{description}

We note that equal treatment is also often referred to as symmetry in the literature.
The uniqueness of the Shapley decomposition in exhibiting the three
described properties is often used as the justification for its application. Furthermore, there are numerous sets of axiomatic properties that lead to the Shapley value decomposition as a solution (see, e.g., \cite{feltkamp1995alternative}). In the statistics literature, it is known that the axioms for decomposition of the coefficient of determination that are proposed by \cite{kruskal1987relative} correspond exactly to the Shapley values (cf. \cite{genizi1993decomposition}).

When conducting statistical estimation and computation, the assumption
of randomness of data necessitates that we address not only the problem
of point estimation, but also variability quantification. In \cite{Huettner:2012aa},
variability for the coefficient of determination Shapley values were
quantified via the use of bootstrap confidence intervals (CIs). Combined
with the usual computational intensiveness of bootstrap resampling
(see, e.g., \cite{Efron:1988aa} and \cite[Ch. 12]{Baglivo:2005aa}),
the combinatory nature of the computation of \eqref{eq: shapley value} 
(notice that $\left|\mathcal{P}\right|=d!$) compounds the time complexity of such a method, which is already of order $\mathcal{O}(2^d)$.
In this article, we seek to provide an asymptotic method for computing CIs for the Shapley values.

Our approach uses the joint asymptotic normality result of the elements
in a correlation matrix, under an elliptical assumption, via \cite{Steiger:1984aa},
combined with asymptotic normality results concerning the determinants
of a correlation matrix, of \cite{Hedges:1983ab} and \cite{Olkin:1995aa}.
Using these results, we derive the asymptotic joint distribution for
the $R^{2}\left(\mathcal{Z}_{n}\right)$ Shapley values, which allows
us to construct CIs for each of the values and their contrasts.
We assess the finite sample properties of our constructions via a comprehensive
Monte Carlo study and demonstrate the use of our CIs via applications
to real estate price data.

The remainder of the article proceeds as follows. 
In Section \ref{sec:main}, we present our main results regarding the asymptotic distribution of the coefficient of determination Shapley values, and their CI constructions.
In Section \ref{s: monte}, we present a comprehensive Monte Carlo study of our CI construction method. 
In Section \ref{sec:realestate}, we demonstrate how our results can be applied to real estate price data. 
Conclusions are lastly drawn in Section \ref{sec:discussion}.
% ---------------------------------------------------------
\section{Main results} \label{sec:main}
% ---------------------------------------------------------
\subsection{The correlation matrix}
% ---------------------------------------------------------
Let \textbf{$\bm{Z}\in\mathbb{R}^{d+1}$ }be a random variable with
mean vector $\bm{\mu}\in\mathbb{R}^{d+1}$ and covariance matrix $\bm{\Sigma}\in\mathbb{R}^{\left(d+1\right)\times\left(d+1\right)}$.
Then, we can define the coefficient of multivariate kurtosis \cite{Mardia:1970aa}
by

\[
\kappa=\frac{1}{\left(d+1\right)\left(d+3\right)}\mathbb{E}\left[\left(\bm{Z}-\bm{\mu}\right)^{\top}\bm{\Sigma}^{-1}\left(\bm{Z}-\bm{\mu}\right)\right]^{2}\text{.}
\]
Let $\rho_{jk}=\text{cor}\left(Z_{j},Z_{k}\right)$, for $j,k\in\left\{ 0\right\} \cup\left[d\right]$
such that $j\ne k$. Assume that $\bm{Z}$ arises from an elliptical
distribution (cf. \cite[Ch. 2]{Fang1990}) and let $\mathcal{Z}_{n}$
be an IID sample with the same distribution as $\bm{Z}$. Then, we
may estimate $\rho_{jk}$ using the sample correlation coefficient
\eqref{eq: sample corr}. Upon writing $\text{acov}$ to denote the
asymptotic covariance, we have the following result due to Corollary
1 of \cite{Steiger:1984aa}.
\begin{lem}
\label{lem steiger}If $\bm{Z}$ arises from an elliptical distribution
and has coefficient of multivariate kurtosis $\kappa$, then the normalized
coefficients of correlation $\zeta_{jk}=n^{1/2}\left(R_{jk}-\rho_{jk}\right)$
($j,k\in\left\{ 0\right\} \cup\left[d\right]$; $j\ne k$) converge
to a jointly normal distribution with asymptotic mean and covariance
elements $0$ and
\begin{eqnarray}
\text{\normalfont acov}\!\left(\zeta_{gh},\zeta_{jk}\right) & = & \kappa\left[\rho_{gh}\rho_{jk}\left(\rho_{gj}^{2}+\rho_{hj}^{2}+\rho_{gk}^{2}+\rho_{hk}^{2}\right)/2+\rho_{gj}\rho_{hk}+\rho_{gk}\rho_{hj}\right]\nonumber \\
 &  & -\kappa\left[\rho_{gh}\left(\rho_{hj}\rho_{hk}+\rho_{gj}\rho_{gk}\right)+\rho_{jk}\left(\rho_{gj}\rho_{hj}+\rho_{gk}\rho_{hk}\right)\right]\text{.}\label{eq: cov of cor}
\end{eqnarray}
\end{lem}
\begin{rem}
We note that the elliptical distribution assumption above can be replaced by a broader pseudo-elliptical assumption, as per \cite{Yuan:1999aa} and \cite{Yuan:2000aa}. 
This is a wide class of distributions that includes some that may not be symmetric. Due to the complicated construction of the class, we refer the interested reader to the source material for its definition.
\end{rem}

\begin{rem}
We may state a similar result that replaces the elliptical assumption by a fourth moments existence assumption instead, using Proposition 2 of \cite{Steiger:1984aa}. In order to make practical use of such an assumption,
we require the estimation of $\left(d+1\right)!$/$\left[\left(d-3\right)!4!\right]$ fourth order moments instead of a single kurtosis term $\kappa$. 
Such a result may be useful when the number of fourth order moments is small, but become infeasible rapidly, as $d$ increases.
\end{rem}
Let $\mathcal{V}\subseteq\left\{ 0\right\} \cup\left[d\right]$, where
$\mathcal{V}=\left\{ v_{1},\dots,v_{\left|\mathcal{V}\right|}\right\} $.
Define $\mathbf{C}_{n}\left(\mathcal{V}\right)$ in the same manner
as \eqref{eq: correlation S}, and let 
\[
\mathbf{R}\left(\mathcal{U}\right)=\left[\begin{array}{cccc}
1 & \rho_{u_{1}u_{2}} & \cdots & \rho_{u_{1}u_{\left|\mathcal{U}\right|}}\\
\rho_{u_{2}u_{1}} & 1 & \cdots & \rho_{u_{2}u_{\left|\mathcal{U}\right|}}\\
\vdots & \vdots & \ddots & \vdots\\
\rho_{u_{\left|\mathcal{U}\right|}u_{1}} & \rho_{u_{\left|\mathcal{U}\right|}u_{2}} & \cdots & 1
\end{array}\right]
\]
and

\[
\mathbf{R}\left(\mathcal{V}\right)=\left[\begin{array}{cccc}
1 & \rho_{v_{1}v_{2}} & \cdots & \rho_{v_{1}v_{\left|\mathcal{V}\right|}}\\
\rho_{v_{2}v_{1}} & 1 & \cdots & \rho_{v_{2}v_{\left|\mathcal{V}\right|}}\\
\vdots & \vdots & \ddots & \vdots\\
\rho_{v_{\left|\mathcal{V}\right|}v_{1}} & \rho_{v_{\left|\mathcal{V}\right|}v_{2}} & \cdots & 1
\end{array}\right]\text{.}
\]
The following theorem is adapted from a result of \cite{Hedges:1983ab}
(also appearing as Theorem 1 in \cite{Hedges:1983aa}). Our result
expands upon the original theorem, to allow for inference regarding
elliptically distributed data, and not just normally distributed data.
We further fix some typographical matters that appear in both \cite{Hedges:1983aa}
and \cite{Hedges:1983ab}.
\begin{lem}
Assume the same conditions as in \eqref{lem steiger}. 
Then, the normalized covariance determinant
$\delta\!\left(\mathcal{U}\right)=n^{1/2}\left(\left|\mathbf{C}_{n}\left(\mathcal{U}\right)\right|-\left|\mathbf{R}\left(\mathcal{U}\right)\right|\right)$
(where $\mathcal{U}$ and $\mathcal{V}$ are nonempty subsets of $\left\{ 0\right\} \mathcal{\cup}\left[d\right]$)
converges to a jointly normal distribution, with asymptotic mean and covariance elements $0$ and
\begin{equation}
\text{\normalfont acov}\!\left(\delta\!\left(\mathcal{U}\right),\delta\!\left(\mathcal{V}\right)\right)=\sum_{g,h\in\mathcal{U}}\sum_{j,k\in\mathcal{V}}r_{\mathcal{U}}\left(g,h\right)r_{\mathcal{V}}\left(j,k\right)\text{\normalfont acov}
\!\left(\zeta_{gh},\zeta_{jk}\right)\text{,}\label{eq: cov det}
\end{equation}
where $\text{\normalfont acov}\!\left(\zeta_{gh},\zeta_{jk}\right)$ is as per \eqref{eq: cov of cor},
\[
%\left[\begin{array}{cccc}
\begin{bmatrix}
r_{\mathcal{U}}\!\left(\,u_{1},u_{1}\right) & r_{\mathcal{U}}\!\left(\,u_{1},u_{2}\right) & \cdots & r_{\mathcal{U}}\!\left(\,u_{1},u_{\left|\mathcal{U}\right|}\right)\\
r_{\mathcal{U}}\!\left(\,u_{2},u_{1}\right) & r_{\mathcal{U}}\!\left(\,u_{2},u_{2}\right) & \cdots & r_{\mathcal{U}}\!\left(\,u_{2},u_{\left|\mathcal{U}\right|}\right)\\
\vdots & \vdots & \ddots & \vdots\\
r_{\mathcal{U}}\!\left(u_{\left|\mathcal{U}\right|},u_{1}\right) & r_{\mathcal{U}}\!\left(u_{\left|\mathcal{U}\right|},u_{2}\right) & \cdots & r_{\mathcal{U}}\!\left(u_{\left|\mathcal{U}\right|},u_{\left|\mathcal{U}\right|}\right)
\end{bmatrix}
%\end{array}\right]
=\left|\mathbf{R}\left(\mathcal{U}\right)\right|\mathbf{R}^{-1}\left(\mathcal{U}\right)\text{,}
\]
and $r_{\mathcal{V}}\left(j,k\right)$ ($j,k\in\mathcal{V}$) is defined
similarly.
\end{lem}
\begin{proof}
The result is due to an application of the delta method (see, e.g.,
\cite[Thm. 3.1]{vanderVaart1998}) and the fact that for any matrix
$\mathbf{R}$, the derivative of its determinant is $\partial\left|\mathbf{R}\right|/\partial\mathbf{R}=\left|\mathbf{R}\right|\mathbf{R}^{-\top}$
\cite[Sec. 17.45]{Seber2008}. Notice that we use the unconstrained
case of the determinant derivative, since we sum over each pair of
coordinates, where $g\ne h$ or $j\ne k$, twice.
\end{proof}
\begin{rem}
If $\mathbf{R}$ is symmetric, then $\partial\left|\mathbf{R}\right|/\partial\mathbf{R}=\left|\mathbf{R}\right|\left[2\mathbf{R}^{-1}-\text{diag}\left(\mathbf{R}^{-1}\right)\right]$
\cite[Sec. 17.45]{Seber2008}. 
Using this fact, we may write \eqref{eq: cov det} in the alternative, and more computationally efficient form
\[
\text{acov}\!\left(\delta\!\left(\mathcal{U}\right),\delta\!\left(\mathcal{V}\right)\right)
=\sum_{\substack{g\le h \\ g,h\in\mathcal{U}}}
\sum_{\substack{j\le k \\ j,k\in\mathcal{V}}}
r_{\mathcal{U}}^{*}\left(g,h\right)r_{\mathcal{V}}^{*}\left(j,k\right)\text{acov}\!\left(\zeta_{gh},\zeta_{jk}\right)\text{,}
\]
where
\begin{eqnarray*}
%\left[\begin{array}{cccc}
\begin{bmatrix}
r_{\mathcal{U}}^{*}\!\left(\,u_{1},u_{1}\right) & r_{\mathcal{U}}^{*}\!\left(\,u_{1},u_{2}\right) & \cdots & r_{\mathcal{U}}^{*}\!\left(\,u_{1},u_{\left|\mathcal{U}\right|}\right)\\
r_{\mathcal{U}}^{*}\!\left(\,u_{2},u_{1}\right) & r_{\mathcal{U}}^{*}\!\left(\,u_{2},u_{2}\right) & \cdots & r_{\mathcal{U}}^{*}\!\left(\,u_{2},u_{\left|\mathcal{U}\right|}\right)\\
\vdots & \vdots & \ddots & \vdots\\
r_{\mathcal{U}}^{*}\!\left(u_{\left|\mathcal{U}\right|},u_{1}\right) & r_{\mathcal{U}}^{*}\!\left(u_{\left|\mathcal{U}\right|},u_{2}\right) & \cdots & r_{\mathcal{U}}^{*}\!\left(u_{\left|\mathcal{U}\right|},u_{\left|\mathcal{U}\right|}\right)
\end{bmatrix}
%\end{array}\right] 
& = &
2\left|\mathbf{R}\left(\mathcal{U}\right)\right|\mathbf{R}^{-1}\left(\mathcal{U}\right)\\
 &   -& \left|\mathbf{R}\left(\mathcal{U}\right)\right|\text{diag}\left(\mathbf{R}^{-1}\left(\mathcal{U}\right)\right)\text{,}
\end{eqnarray*}
and $r_{\mathcal{V}}^{*}\left(j,k\right)$ ($j,k\in\mathcal{V}$)
is defined similarly.
\end{rem}
% ---------------------------------------------------------
\subsection{The coefficient of determination}
% ---------------------------------------------------------
Let \text{plim} denote convergence in probability, so that for any sequence $\{X_n\}$, and any random variable $X$, the statement
\[
\lim_{n \rightarrow \infty} \Pr(|X_n - X| > \varepsilon) = 0  \, ,
\]
for every $\varepsilon > 0$, can be written as $\plim_{n \rightarrow \infty} X_n = X$.

Now, recall definition \eqref{eq: R2 S}, and further let $\rho_{\mathcal{S}}^{2}=\plim_{n\rightarrow \infty}R_{\mathcal{S}}^{2}\left(\mathcal{Z}_{n}\right)$.
We adapt from and expand upon \cite[Thm. 2]{Hedges:1983aa} in the following result. This result also fixes typographical errors that appear in the original theorem, as well as in \cite{Hedges:1983ab}.
\begin{lem}
\label{lem R2}Assume the same conditions as in \cref{lem steiger}.
Then, the normalize coefficient of determination $\lambda\left(\mathcal{S}\right)=n^{1/2}\left(R_{\mathcal{S}}^{2}\left(\mathcal{Z}_{n}\right)-\rho_{\mathcal{S}}^{2}\right)$
(where $S$ and $\mathcal{T}$ are nonempty subsets of $\left[d\right]$)
converges to a jointly normal distribution, with asymptotic mean and
covariance elements $0$ and
\begin{eqnarray}
\text{\normalfont acov}\!\left(\lambda\!\left(\mathcal{S}\right),\lambda\!\left(\mathcal{T}\right)\right) & = & \frac{1}{\left|\mathbf{R}\left(\mathcal{S}\right)\right|\left|\mathbf{R}\left(\mathcal{T}\right)\right|}\text{ \normalfont acov}\!\left(\delta\!\left(\left\{ 0\right\} \cup\mathcal{S}\right),\delta\!\left(\left\{ 0\right\} \cup\mathcal{T}\right)\right)\nonumber \\
 &  & +\frac{\left|\mathbf{R}\left(\left\{ 0\right\} \cup\mathcal{S}\right)\right|\left|\mathbf{R}\left(\left\{ 0\right\} \cup\mathcal{T}\right)\right|}{\left|\mathbf{R}\left(\mathcal{S}\right)\right|^{2}\left|\mathbf{R}\left(\mathcal{T}\right)\right|^{2}}\text{\normalfont acov}\!\left(\delta\!\left(\mathcal{S}\right),\delta\!\left(\mathcal{T}\right)\right)\nonumber \\
 &  & -\frac{\left|\mathbf{R}\left(\left\{ 0\right\} \cup\mathcal{S}\right)\right|}{\left|\mathbf{R}\left(\mathcal{S}\right)\right|^{2}\left|\mathbf{R}\left(\mathcal{T}\right)\right|}\text{\normalfont acov}\!\left(\delta\!\left(\mathcal{S}\right),\delta\!\left(\left\{ 0\right\} \cup\mathcal{T}\right)\right)\nonumber \\
 &  & -\frac{\left|\mathbf{R}\left(\left\{ 0\right\} \cup\mathcal{T}\right)\right|}{\left|\mathbf{R}\left(\mathcal{S}\right)\right|\left|\mathbf{R}\left(\mathcal{T}\right)\right|^{2}}\text{\normalfont acov}\!\left(\delta\!\left(\left\{ 0\right\} \cup\mathcal{S}\right),\delta\!\left(\mathcal{T}\right)\right)\text{.}\label{eq: cov R2}
\end{eqnarray}
\end{lem}
\begin{proof}
We apply the delta method again, using the functional form \eqref{eq: R2 S},
and using the fact that
\[
\frac{\partial}{\partial\left(x,y\right)}\left(1-\frac{x}{y}\right)=\left(-\frac{1}{y},\frac{x}{y^{2}}\right)\text{.}
\]
\end{proof}
\begin{rem}
When $\mathcal{S}=\mathcal{T}$, \eqref{eq: cov R2} yields the usual form for the asymptotic variance of $R_{\mathcal{S}}^{2}\left(\mathcal{Z}_{n}\right)$:
$4\kappa\rho_{\mathcal{S}}^{2}\left(1-\rho_{\mathcal{S}}^{2}\right)^{2}$
(cf. \cite{Yuan:2000aa}).
\end{rem}
%
% ---------------------------------------------------------
\subsection{The Shapley values}
% ---------------------------------------------------------
For every $j\in\left[d\right]$ and $\mathcal{S}\subseteq\mathcal{S}_{j}=\left[d\right]-\left\{ j\right\} $,
there are $\left|\mathcal{S}\right|!\left(d-\left|\mathcal{S}\right|-1\right)!$
elements of $\bm{\pi}\in\mathcal{P}$ such that $\mathcal{S}_{j}\left(\bm{\pi}\right)=\mathcal{S}$.
Thus, we may write
\begin{equation}
V_{j}\left(\mathcal{Z}_{n}\right)=\sum_{\mathcal{S}\subseteq\mathcal{S}_{j}}\omega\left(\mathcal{S}\right)\left[R_{\left\{ j\right\} \cup\mathcal{S}}^{2}\left(\mathcal{Z}_{n}\right)-R_{\mathcal{S}}^{2}\left(\mathcal{Z}_{n}\right)\right]\text{,}\label{eq: alt Shapley}
\end{equation}
where $\omega\left(\mathcal{S}\right)=\left|\mathcal{S}\right|!\left(d-\left|\mathcal{S}\right|-1\right)!/d!$,
and define $v_{j}=\plim_{n\rightarrow \infty}V_{j}\left(\mathcal{Z}_{n}\right)$.
Using this functional form \eqref{eq: alt Shapley}, we may apply the delta method once more, in order to derive the following joint asymptotic normal distribution result regarding the Shapley values $V_{j}\left(\mathcal{Z}_{n}\right)$, for $j\in\left[d\right]$. 
\begin{rem}
The form \eqref{eq: alt Shapley} is a useful computational trick that reduces the computational time of form \eqref{eq: shapley value} and results in more efficient computations for fixed $d$. 
It is unclear whether other formulations such as that of~\cite{hart1988potential} can make the computation time even faster.
Unfortunately, however, there is no formulation that reduces the $\mathcal{O}(2^d)$ scaling, as $d$ increases.
\end{rem}

\begin{thm}
\label{thm main}Assume the same conditions as in \cref{lem steiger}.
Then, the normalized Shaply values $\xi_{j}=n^{1/2}\left(V_{j}\left(\mathcal{Z}_n\right)-v_{j}\right)$
(where $j,k\in\left[d\right]$) converge to a jointly normal distribution,
with asymptotic mean and covariance elements $0$ and $\text{\normalfont acov}\!\left(\xi_{j},\xi_{k}\right)=a_{jk}+b_{jk}-c_{jk}-d_{jk}$.
Here,
\[
a_{jk}=\sum_{\mathcal{S}\subseteq\mathcal{S}_{j}}\sum_{\mathcal{T}\subseteq\mathcal{S}_{k}}\omega\left(\mathcal{S}\right)\omega\left(\mathcal{T}\right)\text{\normalfont acov}\!\left(\lambda\left(\left\{ j\right\} \cup\mathcal{S}\right),\lambda\left(\left\{ k\right\} \cup\mathcal{T}\right)\right)\text{,}
\]
\[
b_{jk}=\sum_{\mathcal{S}\subseteq\mathcal{S}_{j}}\sum_{\mathcal{T}\subseteq\mathcal{S}_{k}}\omega\left(\mathcal{S}\right)\omega\left(\mathcal{T}\right)\text{\normalfont acov}\!\left(\lambda\left(\mathcal{S}\right),\lambda\left(\mathcal{T}\right)\right)\text{,}
\]
\[
c_{jk}=\sum_{\mathcal{S}\subseteq\mathcal{S}_{j}}\sum_{\mathcal{T}\subseteq\mathcal{S}_{k}}\omega\left(\mathcal{S}\right)\omega\left(\mathcal{T}\right)\text{\normalfont acov}\!\left(\lambda\left(\mathcal{S}\right),\lambda\left(\left\{ k\right\} \cup\mathcal{T}\right)\right)\text{,}
\]
and
\[
d_{jk}=\sum_{\mathcal{S}\subseteq\mathcal{S}_{j}}\sum_{\mathcal{T}\subseteq\mathcal{S}_{k}}\omega\left(\mathcal{S}\right)\omega\left(\mathcal{T}\right)\text{\normalfont acov}\!\left(\lambda\left(\left\{ j\right\} \cup\mathcal{S}\right),\lambda\left(\mathcal{T}\right)\right)\text{,}
\]
where $\lambda\left(\mathcal{S}\right)$ is as defined in \cref{lem R2},
for nonempty subsets $\mathcal{S}\subseteq\left[d\right]$.
\end{thm}
Using the result above, we may apply the delta method again in order
to construct asymptotic CIs or hypothesis tests regarding any continuous
function of the $d$ Shapley values for the coefficient of determination.
Of particular interest is the asymptotic CI for each of the individual
Shapley values and the hypothesis test for the difference between
two Shapley values.

The asymptotic $100\left(1-\alpha\right)\%$ CI for the $j\text{th}$
expected Shapley value $v_{j}$ has the usual form
\[
\left(V_{j}\left(\mathcal{Z}_{n}\right)-\Phi^{-1}\left(1-\frac{\alpha}{2}\right)\sqrt{\frac{\text{avar}\left(\xi_{j}\right)}{n}},V_{j}\left(\mathcal{Z}_{n}\right)+\Phi^{-1}\left(1-\frac{\alpha}{2}\right)\sqrt{\frac{\text{avar}\left(\xi_{j}\right)}{n}}\right)\text{,}
\]
where $\text{avar}\left(\xi_{j}\right)=\text{\normalfont acov}\!\left(\xi_{j},\xi_{j}\right)$
denotes the asymptotic variance of $\xi_{j}$ and $\Phi^{-1}$ is
the inverse cumulative distribution function of the standard normal
distribution. The $z$-statistic for the test of the null and alternative
hypotheses

\begin{equation}
    \text{H}_{0}:v_{j} = v_{k}
    \text{ and }
    \text{H}_{1}:v_{j}\ne v_{k}\text{,}
    \label{eq: two shapley test}
\end{equation}
for $j,k\in\left[d\right]$ such that $j\ne k$, is

\[
\Delta_{n}=\frac{\sqrt{n}\left(V_{j}\left(\mathcal{Z}_{n}\right)-V_{k}\left(\mathcal{Z}_{n}\right)\right)}{\sqrt{\text{avar}\left(\xi_{j}\right)+\text{avar}\left(\xi_{k}\right)-2\text{\normalfont acov}\!\left(\xi_{j},\xi_{k}\right)}}\text{,}
\]
where $\Delta_{n}$ has an asymptotic standard normal distribution.
\begin{rem}
In practice, we do not know the necessary elements $\kappa$ and $\rho_{jk}$
$\left(j,k\in\left\{ 0\right\} \cup\left[d\right];j\ne k\right)$
that are required in order to specify the asymptotic covariance terms
in \cref{lem steiger}--\cref{lem R2} and \cref{thm main}.
However, by Slutsky's theorem, we have the usual result that any $\text{acov}$
(or $\text{avar}$) term can be replaced by the estimator $\widehat{\text{acov}}_{n}$
(or $\widehat{\text{avar}}_n$), which replaces $\kappa$ by the estimator
of \cite{Mardia:1970aa}
\[
\hat{\kappa}\left(\mathcal{Z}_{n}\right)=\frac{\sum_{i=1}^{n}\left[\left(\bm{Z}_{i}-\bar{\bm{Z}}\right)^{\top}\hat{\bm{\Sigma}}_{n}^{-1}\left(\left\{ 0\right\} \cup\left[d\right]\right)\left(\bm{Z}_{i}-\bar{\bm{Z}}\right)\right]^{2}}{n\left(d+1\right)\left(d+3\right)}\text{,}
\]
and replaces $\rho_{jk}$ by $R_{jk}\left(\mathcal{Z}_{n}\right)$,
where $\bar{\bm{Z}}^{\top}=\left(\bar{Z}_{0},\dots,\bar{Z}_{d}\right)$. Here,
\[
\hat{\bm{\Sigma}}_{n}\left(\mathcal{\mathcal{U}}\right)=\left[\begin{array}{cccc}
\hat{\Sigma}_{u_{1}u_{1}}\left(\mathcal{Z}_{n}\right) & \hat{\Sigma}_{u_{1}u_{2}}\left(\mathcal{Z}_{n}\right) & \cdots & \hat{\Sigma}_{u_{1}u_{\left|\mathcal{U}\right|}}\left(\mathcal{Z}_{n}\right)\\
\hat{\Sigma}_{u_{2}u_{1}}\left(\mathcal{Z}_{n}\right) & \hat{\Sigma}_{u_{2}u_{2}}\left(\mathcal{Z}_{n}\right) & \cdots & \hat{\Sigma}_{u_{2}u_{\left|\mathcal{U}\right|}}\left(\mathcal{Z}_{n}\right)\\
\vdots & \vdots & \ddots & \vdots\\
\hat{\Sigma}_{u_{\left|\mathcal{U}\right|}u_{1}}\left(\mathcal{Z}_{n}\right) & \hat{\Sigma}_{u_{\left|\mathcal{U}\right|}u_{2}}\left(\mathcal{Z}_{n}\right) & \cdots & \hat{\Sigma}_{u_{\left|\mathcal{U}\right|}u_{\left|\mathcal{U}\right|}}\left(\mathcal{Z}_{n}\right)
\end{array}\right],
\]
where $\hat{\Sigma}_{jk}\left(\mathcal{Z}_{n}\right)$ is the sample covariance between the $j$th and $k$th dimension of $\bm{Z}$ ($j,k\in\{0\}\cup[d]$). For example, the estimated test statistic
\begin{equation}\label{eq:delta_n}
\hat{\Delta}_{n}=\frac{\sqrt{n}\left(V_{j}\left(\mathcal{Z}_{n}\right)-V_{k}\left(\mathcal{Z}_{n}\right)\right)}{\sqrt{\widehat{\text{avar}}_n\left(\xi_{j}\right)+\widehat{\text{avar}}_n\left(\xi_{k}\right)-2\widehat{\text{acov}}_n\left(\xi_{j},\xi_{k}\right)}}\text{,}
\end{equation}
for the hypotheses \eqref{eq: two shapley test}, retains the property
of having an asymptotically standard normal distribution.
\end{rem}
% ---------------------------------------------------------
\section{Monte Carlo studies and benchmarks}\label{s: monte}
% ---------------------------------------------------------
In each of the following three Monte Carlo studies, we simulate a large number $N$ of random samples $\mathcal{Z}_n^{(i)}, \; i \in [N]$, of size $n$, from a chosen distribution $\mathcal{D}$.
%so that $\mathbf{Z} \overset{\text{sim}}{\sim} \mathcal{D}$
For each sample, we apply \cref{lem steiger} to calculate an asymptotic 95\% CI for the first Shapley value $v_1$, producing a set of $N$ observed intervals $\mathcal{I}_N = \left\{[\ell_i, u_i]\, :\, i \in [N]\right\}$, as realisations of the CI $[L,U]$ for $v_1$. The coverage probability $\text{covr}_{v_1}([L,U]) := \Pr (v_1 \in [L,U])$ is then estimated as the proportion of intervals in $\mathcal{I}_N$ containing the population Shapley value
\begin{equation} \label{eq: est coverage}
\widehat{\text{covr}}_{v_1}([L,U]) = \frac{|\{[\ell_i, u_i] \in \mathcal{I}_N \, :\, v_1^{} \in [\ell_i, u_i]\}|}{N}.
\end{equation}
Here, the population Shapley value $v_1$ has the form:
%is calculated via application of Slutsky's theorem in \cref{eq: alt Shapley},
\[
v_1 = %\plim_{n\rightarrow \infty}V_1\left(\mathcal{Z}_{n}\right) =
\sum_{\mathcal{S}\subseteq\mathcal{S}_{1}}\omega\left(\mathcal{S}\right)\left[R_{\left\{ 1\right\} \cup\mathcal{S}}^{2}-R_{\mathcal{S}}^{2}\right],
\]
where $R_{\mathcal{S}}^{2}$ is defined by replacing $\mathbf{C}_n$ in \eqref{eq: R2 S} by the known population correlation matrix $\text{cor}(\mathbf{Z}),$ as determined by the chosen distribution $\mathcal{D}$. In Studies A and B, this population correlation matrix is the $(d+1)\times(d+1)$ matrix $\mathbf{\Sigma}$, with diagonal elements equal to 1 and off-diagonal elements equal to a constant correlation $c \in [0,1)$. That is,
\begin{equation} \label{eq: sigma}
\mathbf{\Sigma} = c\,\mathbf{J}_{d+1} + (1-c)\,\mathbf{I}_{d+1},
%\begin{bmatrix}
%         1      & c      & \ldots & c \\
%         c      & 1      & \ldots & c \\
%         \vdots & \vdots & \ddots & \vdots \\
%         c      & c      & \ldots & 1
%         \end{bmatrix}.
\end{equation}
where $\mathbf{J}_{d+1}$ denotes a 
$(d+1)\times(d+1)$ matrix with all entries equal to $1$. Note that, for fixed variance, larger magnitudes of $c$ map to larger regression coefficients. Thus, these simulations can be viewed as assigning equal regression coefficients to each covariates that are increasing for increasing values of $c$.

In Study C, we are concerned with covariance matrices with off-diagonal elements deviating from $c$. We aim to capture a case where the off-diagonal elements of $\text{cor}(\mathbf{Z})$ are not uniform, and where some may be negative. This is achieved by sampling $\mathcal{Z}_n^{(i)}$ from a multivariate normal distribution $\mathcal{D}_i$, with a symmetric positive definite covariance matrix $\mathbf{\Sigma}_i$ that is sampled at random from a Wishart distribution with scale matrix $\mathbf{\Sigma}$; see \Cref{sec: mc-c}. Accordingly, the population Shapley value $v_1^{(i)}$ is unique to sample $i$, and thus we adjust the coverage estimator $\widehat{\text{covr}}_{v_1}([L,U])$ by replacing $v_1$ on the RHS of \ref{eq: est coverage} by $v_1^{(i)}$. 

%This approach is intended to reduce the over- specialization of our results to uniform and non-negative off-diagonal elements in $\text{cor}(\mathbf{Z})$. 
 To accompany our estimates of $\widehat{\text{covr}}_{v_1}([L,U])$, we also provide Clopper-Pearson CIs \cite{10.2307/2331986} for the coverage probability. We also report the average CI widths, and middle 95\% percentile intervals for the widths. 
 For comparison, we estimate coverage probabilities of non-parametric bootstrap confidence intervals in each of the three studies. 
 To obtain the bootstrap CIs, we set some large number $N_b$ and take random resamples $\mathcal{R}_n^{(r)}, r \in [N_b]$, of size $n$, with replacement, from $\mathcal{Z}^{(i)}_n$. From these resamples, we calculate the set of estimated Shapley values 
 $\mathcal{L}_i = \left\{V_{1}\left(\mathcal{R}_n^{(r)}\right) \, :\, r \in [N_b]\right\}$. The $i$th 95\% bootstrap CI is then taken as the middle 95\% percentile interval of $\mathcal{L}$, and the coverage is estimated as in \ref{eq: est coverage}. 

To obtain results for each pair $(n,c)\in\mathcal{N}\times \mathcal{C}$, where $\mathcal{N} = \{5,10,\ldots,50\}\cup\{100,200,\ldots,2000\}$ and $\mathcal{C} = \{0,0.1,0.2,0.3,0.6,0.9,0.99\}$, we performed $30\times7 = 210$ simulations for each of the three studies. 
We use $N = N_b = 1000$ and $d = 3$, in all cases.
% ---------------------------------------------------------
\subsection{MC Study A}\label{sec: mc-a}
% ---------------------------------------------------------
Here, we choose $\mathcal{D} = N_{d+1}(\mathbf{0},\mathbf{\Sigma})$, so that each sample $\mathcal{Z}^{(i)}_n,$ for $i \in [N],$ is drawn from a multivariate normal distribution, with covariance $\mathbf{\Sigma}$ given in \eqref{eq: sigma}. 

The simulation results in \Cref{fig: A-large} (and in \Cref{fig: A-small}, \Cref{appendix:figures}) show very similar coverage and width performance between the two assessed CIs for moderate and high correlations $c > 0.3$. For lower correlations $c \leq 0.2$, coverage convergence appears to be slower in $n$ than the bootstrap CI for large sample sizes ($n \geq 100$). 
The opposite trend seems to hold for small sample sizes ($n \leq 50$), see the discussion under MC Study C. Also, for the highest correlation $c = 0.99$, coverage performance of the asymptotic CI is overall slightly better than the bootstrap CI.

\begin{figure}[H]
    \begin{center}
        \includegraphics[width = \textwidth]{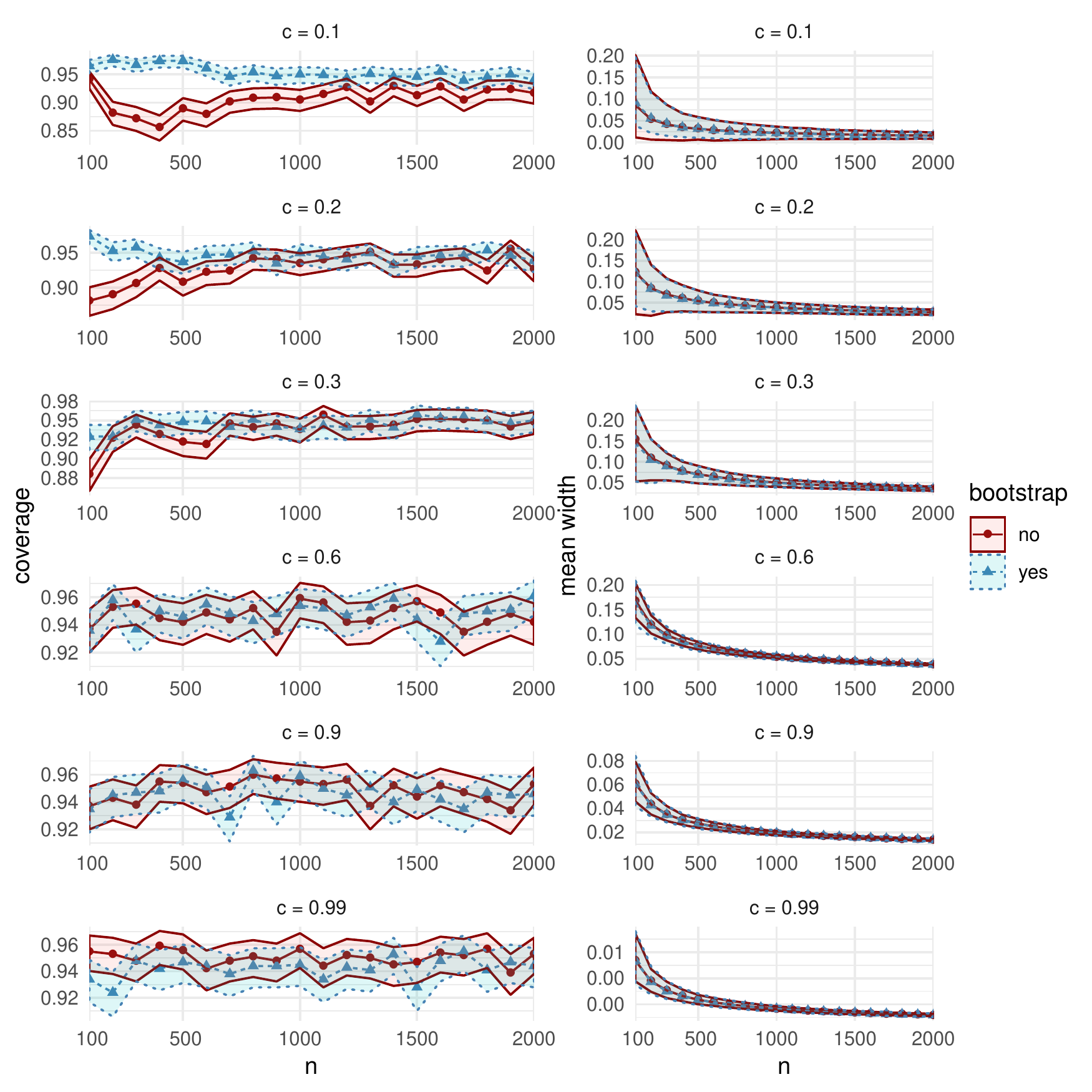}
    \end{center}
    \caption{Comparisons of coverage (left column) and mean width (right column), between the bootstrap CIs (dashed lines with triangular markers) and the asymptotic CIs (solid lines with circular markers) in MC Study A. Rows represent correlations $c$, increasing from $0.1$ on the top row to $0.99$ on the bottom row. The horizontal axes display the sample sizes $n = 100,200,\ldots,2000$.}
    \label{fig: A-large}
\end{figure} 

\subsection{MC Study B}
Here, we choose $\mathcal{D} = t_\nu(\mathbf{0},\mathbf{\Sigma})$ where $t_\nu(\mathbf{\mu},\mathbf{\Sigma})$ is the multivariate Student $t$ distribution with $\nu\in(0,\infty)$ degrees of freedom, mean vector $\mathbf{\mu},$ and scale matrix $\mathbf{\Sigma}.$ Specifically, the $i$th sample $\mathcal{Z}^{(i)}_n$, for $i \in [N]$, we set $\nu = 100$ degrees of freedom, and set $\mathbf{\Sigma}$ as the $(d+1)\times(d+1)$ covariance matrix in \ref{eq: sigma}.

For all sample sizes $n$ and correlations $c$, coverage and width performances are similar to MC Study A (see \Cref{fig: B-small} and \Cref{fig: B-large} in \Cref{appendix:figures}). Of particular interest, in both MC Studies A and B (but not in MC Study C), we observe that for $c = 0$, the estimated coverage probability of the asymptotic CI is almost equal to 1, for all sample sizes greater than 10, while the corresponding bootstrap CIs have estimated coverage equal to 0 (\Cref{fig: B-small-c0} left). Despite this, the average CI widths, though large under small samples, are somewhat smaller than those for bootstrap (\Cref{fig: B-small-c0} right).

\begin{figure}[H]
    \begin{center}
  \includegraphics[width = \textwidth]{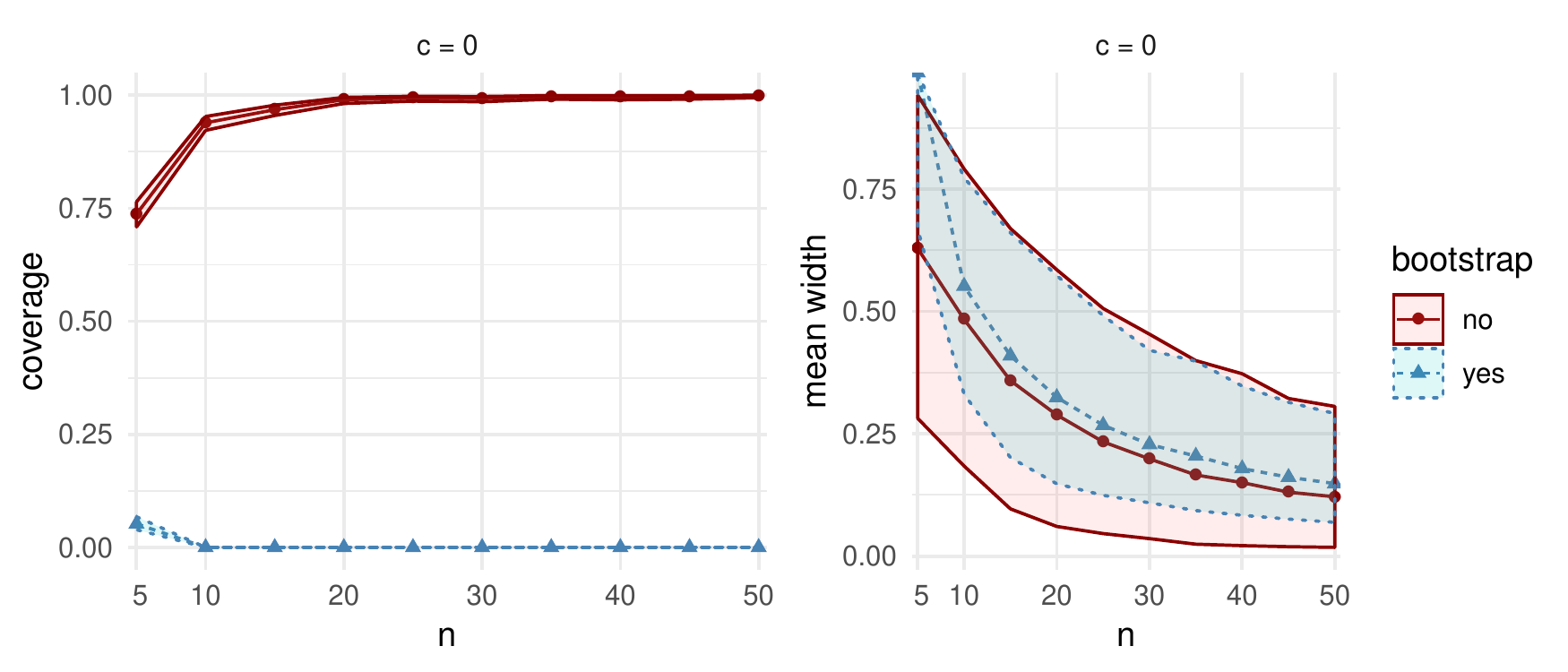}
    \end{center}
    \caption{Comparisons of coverage (left column) and mean width (right column) between bootstrap CIs (dashed lines with triangular markers) and the asymptotic CIs (solid lines with circular markers) in MC Study B, for correlation $c = 0$. The horizontal axes display the sample sizes $n = 5,10,\ldots,50$.}
    \label{fig: B-small-c0}
\end{figure} 

\subsection{MC Study C} \label{sec: mc-c}
Here, we set $\mathcal{D}_i = N_{d+1}(0,\mathbf{\Sigma}_i)$, so that the sample $\mathcal{Z}^{(i)}_n$, for $i \in N$, is drawn from a multivariate normal distribution, 
with covariance matrix $\mathbf{\Sigma}_i$ realised from a Wishart distribution $W_{d+1}(\mathbf{\Sigma}, \nu)$ with scale matrix $\mathbf{\Sigma}$ and $\nu$ degrees of freedom. 
This set up is different from Studies A and B in that the distributions $\mathcal{D}_i$, and therefore the population Shapley values $v_1^{(i)}$, are allowed to differ between samples. 

The distribution $W_{d+1}(\mathbf{\Sigma}, \nu)$ can be understood as the distribution of the sample covariance matrix of a sample of size $\nu+1$ from the distribution $N_{d+1}(0,\mathbf{\Sigma})$  
(cf \cite{fujikoshi2011}). 
This implies that each covariance matrix $\mathbf{\Sigma}_i$ can have non-uniform and negative off-diagonal elements, with variability between the off-diagonal elements increasing as $\nu$ decreases. 
For this study, we set $\nu = 100$. 

Aside for the case $c = 0$, coverage and width statistics are again similar to MC Studies A and B, for all $n$ and $c$ (see \Cref{appendix:figures}). Interestingly, in all three Studies, for small sample sizes ($n \leq 50$), coverage is often higher than for the bootstrap CI, with slightly smaller average widths, as seen in \Cref{fig: C-small} (and in \Cref{fig: A-small} and \Cref{fig: B-small} in \Cref{appendix:figures}).
For the $c = 0$ case, the observed behaviour differs from MC Studies A and B, with bootstrap performing comparatively well for large sample sizes (\Cref{fig: C-small-c0}).

\begin{figure}[H]
    \begin{center}
        \includegraphics[width = \textwidth]{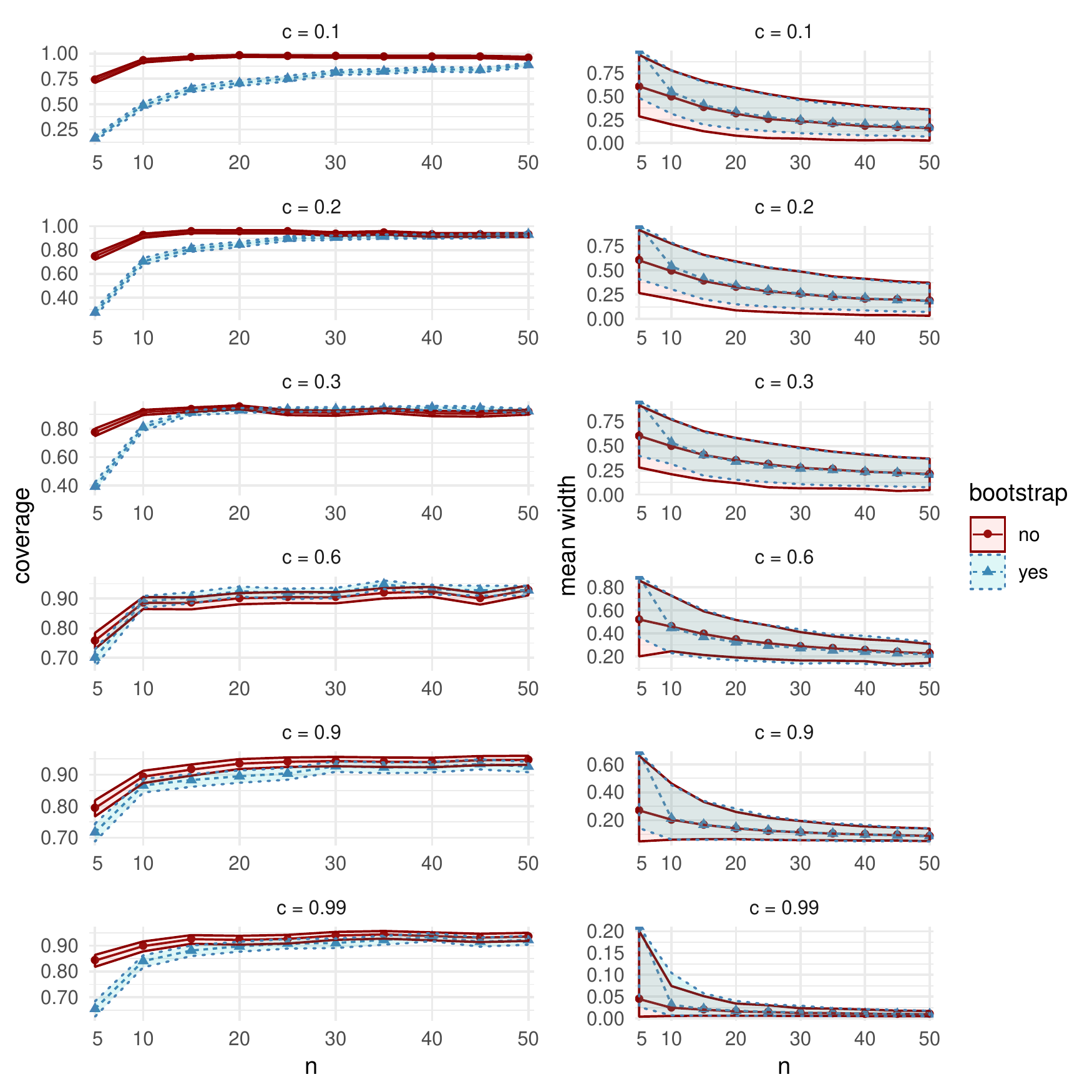}
    \end{center}
    \caption{Comparisons of coverage (left column) and mean width (right column), between the bootstrap CIs (dashed lines with triangular markers) and the asymptotic CIs (solid lines with circular markers) in MC Study C. Rows represent  correlations $c$, increasing from $0.1$ on the top row to $0.99$ on the bottom row. The horizontal axes display the sample sizes $n = 5,10,\ldots,50$.}
    \label{fig: C-small}
\end{figure} 

\begin{figure}[H]
    \begin{center}
  \includegraphics[width = \textwidth]{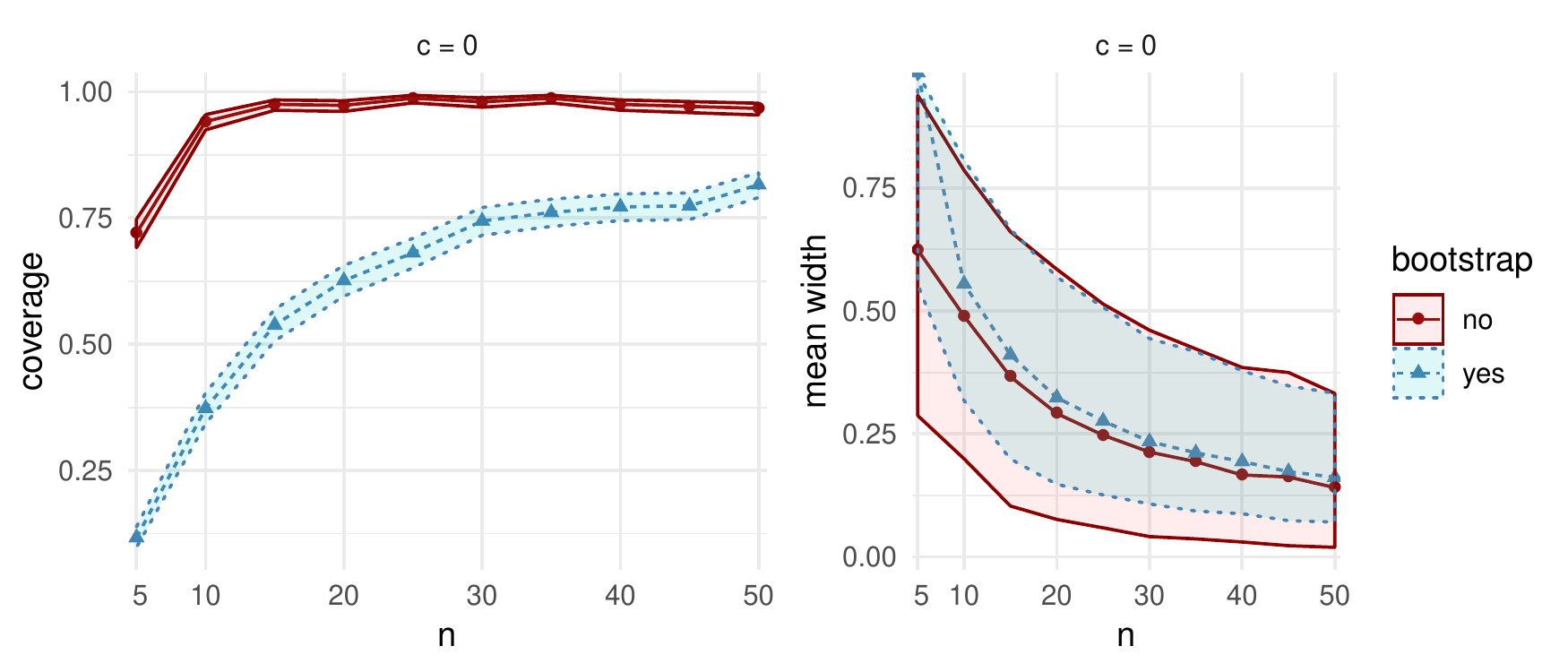}
    \end{center}
    \caption{Comparisons of coverage (left column) and mean width (right column) between bootstrap CIs (dashed lines with triangular markers) and the asymptotic CIs (solid lines with circular markers) in MC Study C, for correlation $c = 0$. The horizontal axes display the sample sizes $n = 5,10,\ldots,50$.}
    \label{fig: C-small-c0}
\end{figure} 

\begin{figure}[H]
    \begin{center}
  \includegraphics[width = \textwidth]{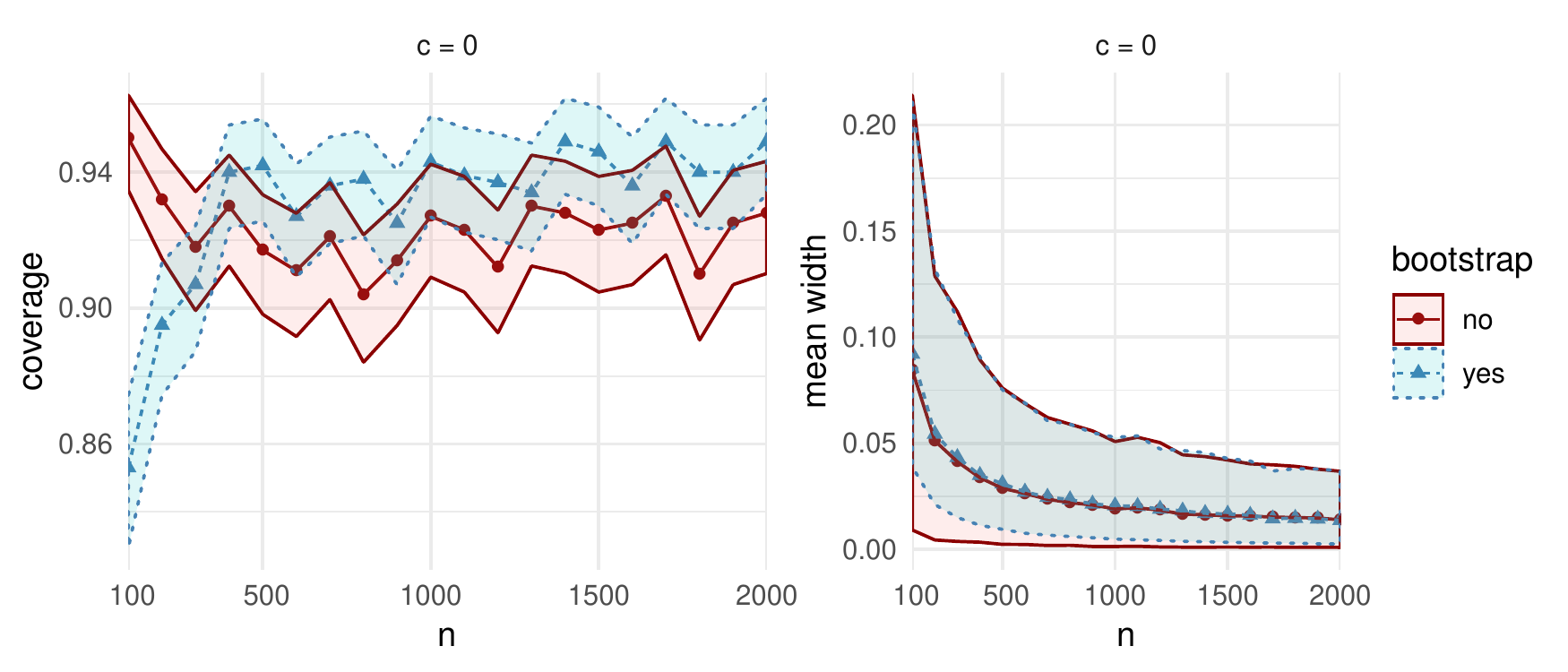}
    \end{center}
    \caption{Comparisons of coverage (left column) and mean width (right column) between bootstrap CIs (dashed lines with triangular markers) and the asymptotic CIs (solid lines with circular markers) in MC Study C, for correlation $c = 0$. The horizontal axes display the sample sizes $n = 100,200,\ldots,2000$.}
    \label{fig: C-large-c0}
\end{figure}
% ---------------------------------------------------------
\subsection{Computational benchmarks}
% ---------------------------------------------------------
From \Cref{fig: benchmark}, we see that the memory usage (left) and mean execution time (right) for the bootstrap CIs are both higher than that for the asymptotic CIs, and that the ratio increases with sample size. 
As $n$ increases, asymptotic CIs become increasingly efficient, compared to the bootstrap CIs. 
On the other hand, as $d$ increases, with $n$ fixed, we expect an increase in the relative efficiency of the bootstrap, since the complexity of calculating $\text{acov}(\xi_j, \xi_k)$ in \cref{thm main} grows faster in $d$ than the complexity of the bootstrap procedure.

\begin{table}[!h]
    \centering
    \caption{Parameters for computational benchmarking.}
    \begin{tabular}{ll}
        \toprule
         Parameter & Value(s)  \\
         \midrule
          Number of features ($d$) &  3\\
          Sample sizes ($n$) & 
          1000, 5000, 10000 \\
          Number of bootstrap resamples ($N_b$) &
          1000 \\
          Number of simulation repetitions ($N$) &
          1000 \\
         \bottomrule
    \end{tabular}
    \label{tab:bechmark}
\end{table}{}

\begin{figure}[!h]
    \begin{center}
         \includegraphics[width=\textwidth]{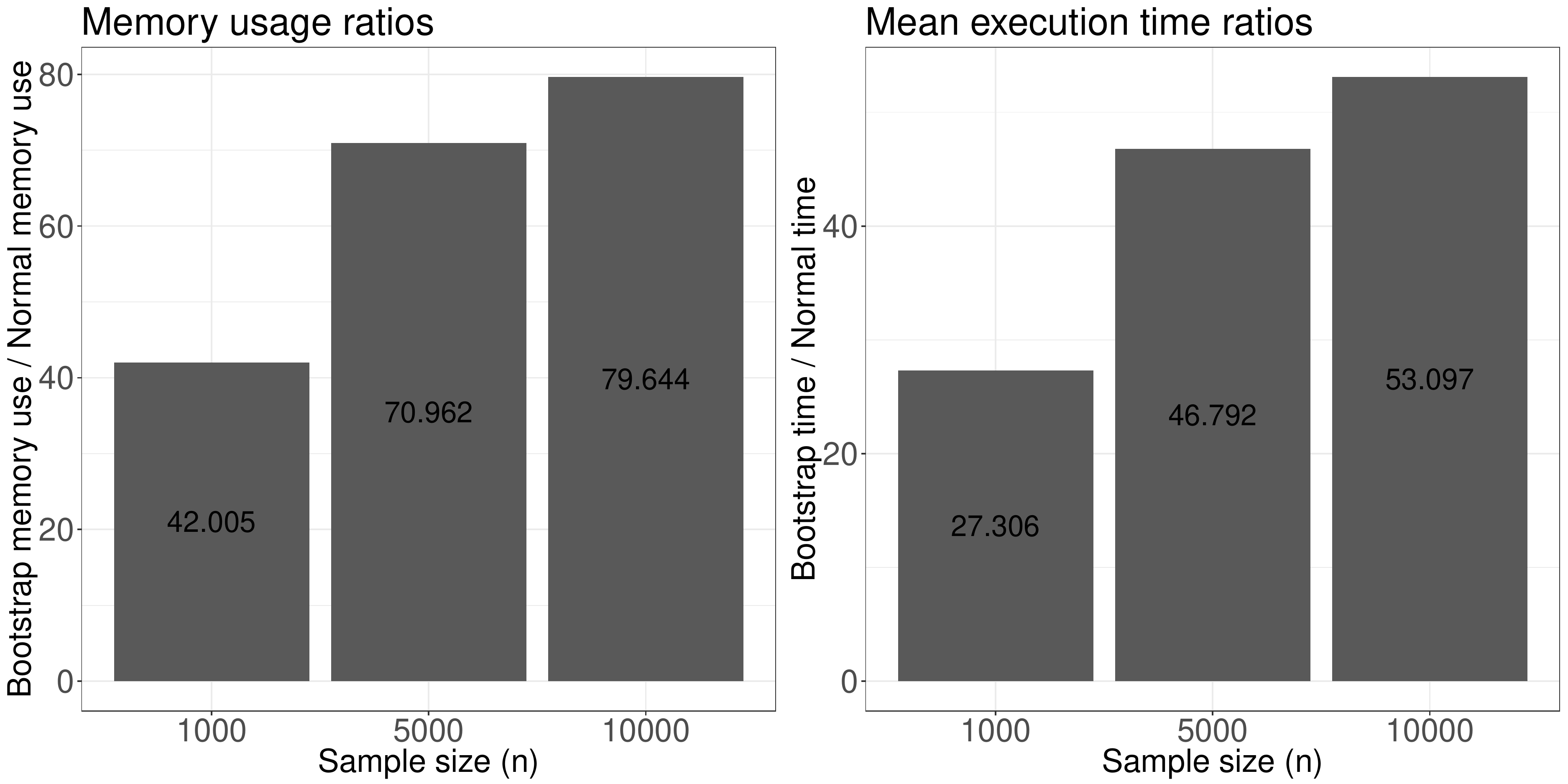}
    \end{center}
    \caption{Computational benchmark metric ratios of confidence interval estimation using the na{\"\i}ve bootstrap over the asymptotic normality approach.}
    \label{fig: benchmark}
\end{figure}
% ---------------------------------------------------------
\subsection{Summary of results and recommendations for use}\label{ss: recommendations}
% ---------------------------------------------------------
Below follows a summary of the general tendencies and observations from our results, and recommendations regarding when to use the asymptotic CIs.
\begin{itemize}
  % General coverage statements of asymptotic intervals
  \item For all correlations $c$ in all three Studies, 
  the estimated coverage probability of asymptotic intervals 
  is above $0.85$ for all sample sizes $n \geq 10$.
  \item For smaller correlations and sample sizes, in particular $c \geq 0.2$ and $n > 15$, the lower bound of the confidence interval for coverage never drops below $0.85$.
  \item For all correlations $c \geq 0.3$ and sample sizes $n > 100$, the lower 
  bound of the confidence interval for coverage 
  never drops below $0.9$.
  \item For small correlations, in particular $c \leq 0.1$ and sample sizes $10 \leq n \leq 100$, the lower 
  bound of the confidence interval for the coverage of the asymptotic CIs  never drops below $0.91$.
  \item For $c = 0$ and $n \geq 15$, the lower bound of the asymptotic CI for coverage never drops below $0.95$ in Studies A and B,  while in Study C the lower bound is at least $0.88$.
  % Coverage vs c plot and width vs c plot
  \item For sample sizes $15 \leq n \leq 50$, the coverage of the asymptotic CI tends to be higher when $c$ is closer to the boundaries of $[0,1]$, as shown in \Cref{fig: cov-vs-cor-small}. 
  \item The average asymptotic CI width is lower when $c$ is nearer to the boundaries of $[0,1]$, see \Cref{fig: CI-widths}.
\end{itemize}

\begin{figure}[!h]
    \begin{center}
       \includegraphics[width=\textwidth]{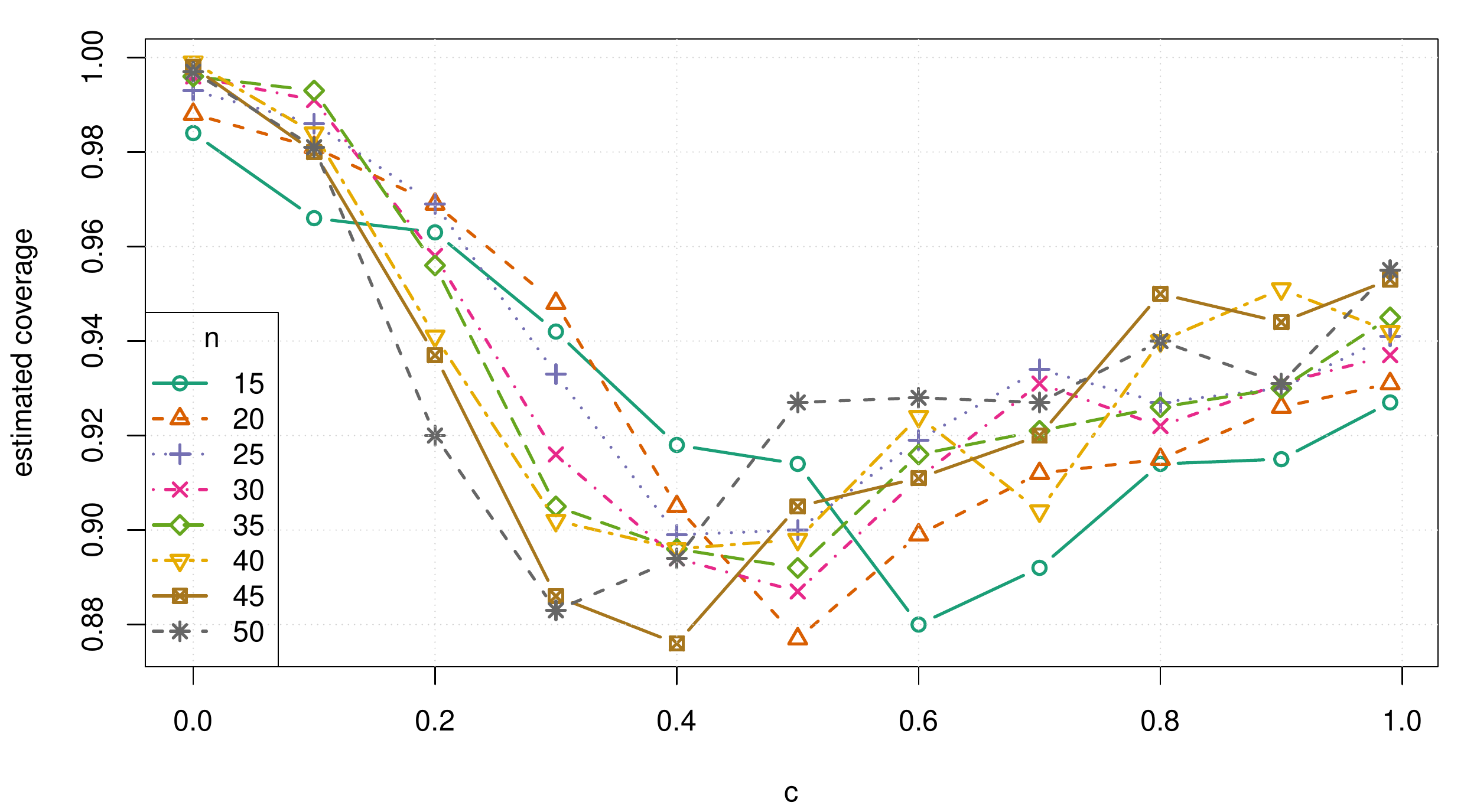}
    \end{center}
    \caption{Estimated coverage probability versus correlation $c$, for small samples sizes $n = 15, 20, \ldots, 50$, in MC Study A. The same patterns can be observed for Studies B and C.}
    \label{fig: cov-vs-cor-small}
\end{figure}

\begin{figure}[!h]
    \begin{center}
        \includegraphics[width=\textwidth]{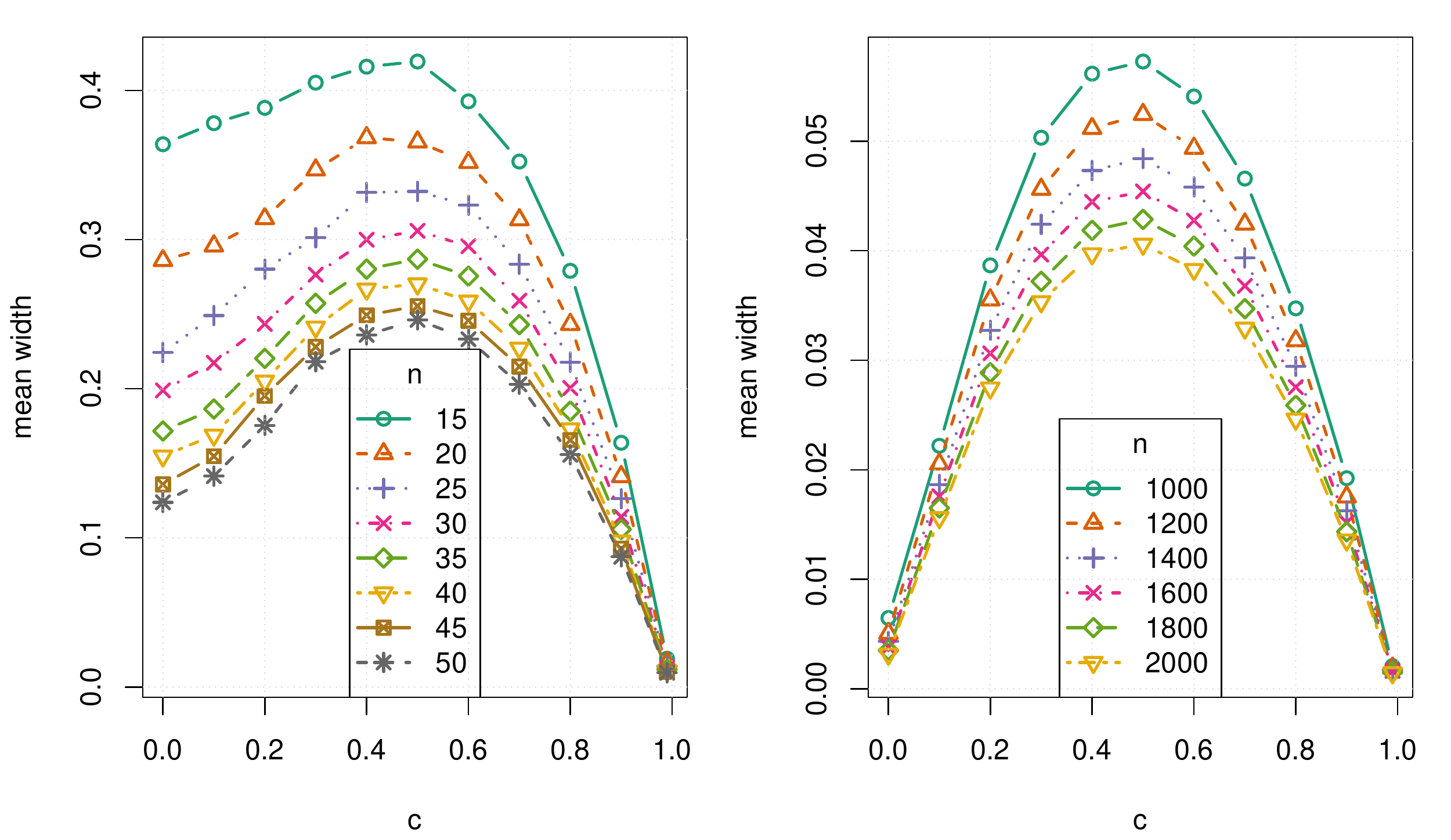}
    \end{center}
    \caption{Average confidence interval width versus correlation $c$, for small sample sizes $n = 15, 20, \ldots, 50$ (left) and large sample sizes $n = 1000, 1200, \ldots, 2000$ (right), in MC Study A. The same patterns are present in MC Study B and C.}
    \label{fig: CI-widths}
\end{figure}

We now make some general observations which apply to all three Studies.
As sample size increases, the estimated coverage initially increases rapidly,
as can be seen, for example, in the left column of \Cref{fig: C-small}. 
For small sample sizes between $n = 5$ and $n = 50$, the asymptotic CIs
typically outperform bootstrap CIs, especially when $c$ lies farther
from $0.5$; there is a clear drop in coverage as $c$
approaches $0.5$, for small samples, as can be seen in \Cref{fig: cov-vs-cor-small}. 
In many cases, the estimated coverage is above $0.9$ for $n \ge 10$. 
However, empirical coverage does not appear to be an increasing function of sample size in general. 
On the top row in the left column of \Cref{fig: A-large}, we observe one example of a clear 
and extended dip in coverage for $n$ in $[100,1000]$.
This gives rise to the general observation that the asymptotic intervals have preferable
coverage statistics over bootstrap for small samples, but not for a certain range of large samples, depending on $c$.

We further observe that, for all $n$, there is a general increase in the average CI width as $c$ approaches $0.5$ from either direction, as in \Cref{fig: CI-widths}.
In all Studies, over all sample sizes and correlations, the bootstrap CI average widths were smaller than the asymptotic CI widths by at most $0.0289$, and vice versa by at most $0.0667$. In general, the asymptotic intervals display favourable widths, though less so near $c = 0.5$.

Based on these observations, we recommend using asymptotic CIs over bootstrap CIs under the following conditions:
\begin{enumerate}[label={(\roman*)}]
    \item Computational time is relevant (e.g., estimating a large number of Shapley values).
    \item The sample size is small (e.g., $n \le 50$).
    \item The correlation between explanatory variables and the response variable is expected to be beyond $\pm0.2$ from $0.5$, or when this is where the highest precision is desired.
\end{enumerate}
We note that our observations are made from an incomplete albeit comprehensive set of simulation scenarios. There are of course an infinite number of combinations of simulation cases and thus we cannot guarantee that our observation applies to all possible DGPs.
% ---------------------------------------------------------
\section{Application: Melbourne real estate and COVID-19}\label{sec:realestate}
% ---------------------------------------------------------
For an interesting application of our methods, in this section we identify significant changes in the behaviour of the local real estate market in Melbourne, Australia, within the period from 1 February to 1 April, between the years 2019 and 2020. In 2020, this corresponds to an early period of growing public concern regarding the novel coronavirus COVID-19. We obtain the Shapley decomposition of the coefficient of multiple correlation $R^{2}$ between observed house prices and a number of property features. We also find significant differences in behaviour between real estate \textit{near} and \textit{far} from the Central Business District (CBD), where \textit{near} is defined to be within $\unit[25]{km}$ of Melbourne CBD, and \textit{far} is defined as non-\textit{near} (see~\Cref{fig:suburb_map}). Note that the nature of this investigation is exploratory and expository; our conclusions are not intended to be taken as evidence for the purposes of policy or decision making.
\\

On 1 February the Australian government announced a temporary ban on foreign arrivals from mainland China, and by 1 April a number of social distancing measures were in place. 
We scraped real estate data from the AUHousePrices website (\url{https://www.auhouseprices.com/}),
to obtain a data set of $13{,}291$ (clean) house sales between 1 January and 18 July in 2019 and 2020. 
We then reduced this date range to capture only the spike in sales observed between 1 February and 1 April (see \Cref{fig:dates_barplot}), giving a remaining sample size of $5110$, which was partitioned into the four subgroups in \Cref{table:subgroups}. Within each of the four subgroups we perform a Yeo-Johnson transformation \cite{10.2307/2673623} to reduce any violation of the assumption of joint pseudo-ellipticity. 
\begin{table}[!h]
    \centering
    \caption{The four subgroups and their sample sizes after partitioning by distance (where \textit{near} $:=$ within $\unit[25]{km}$ of CBD), and year of sale in the period 1 February to 1 April 2019 and 2020.}
    \begin{tabular}{lllll}
        \toprule
         Subgroup: & \textit{near} (2019) & \textit{far} (2019) & \textit{near} (2020) & \textit{far} (2020)  \\
         \midrule
         Sample size:   & 1203 &  953 & 1824 & 1130 \\
         \bottomrule
    \end{tabular}
    \label{table:subgroups}
\end{table}{}

We decompose $R^2$ amongst the covariates: distance to CBD (\textit{CBD}); images used in advertisement (\textit{images}); property land size (\textit{land}); distance to nearest school (\textit{school}); distance to nearest station (\textit{station}); and number of bedrooms $+$ bathrooms $+$ car ports (\textit{room}); along with the response variable, house sale price (\textit{price}). We expect the room covariate to act as a proxy for house size. Thus we decompose $R^2$ for the linear model,
\begin{equation} \label{eq:realestate_model}
    \mathbb{E}(\text{price}) =  \beta_0 + \bm{\beta} (\text{CBD},\, \text{images},\, \text{land},\, \text{school},\, \text{station},\, \text{room})^T, \;\; \beta_0 \in \mathbb{R},\,  \bm{\beta} \in \mathbb{R}^6.
\end{equation}

Fitted to each of the four subgroups, we obtain $R^2 = 0.37$ for model \eqref{eq:realestate_model}, for each subgroup except for the \textit{near} (2020) subgroup, for which $R^2 = 0.39$. The resulting Shapley values and 95\% confidence intervals are listed in~\Cref{table:realestate_results}, and shown graphically in \Cref{fig:realestate_shapley}. From those results we make the following observations regarding attributions of the total variability in house prices explained by model \eqref{eq:realestate_model}:
\begin{enumerate}[label={(\roman*)}]  
    \item \label{en:re1} In both 2019 and 2020, the attribution for distance to CBD was significantly higher for house sales \textit{near} to the CBD, compared to house sales farther from the CBD. Correspondingly, the attribution for roominess was significantly lower for house sales \textit{near} to the CBD.
    \item \label{en:re2} Amongst sales that were \textit{near} to the CBD, distances to the CBD received significantly greater attribution than both land size and roominess in 2019. Unlike the case for sales that were far from the CBD, these differences remained significant in 2020.
    \item \label{en:re3} Distances to stations and schools, as well as images used in advertising, have apparently had little overall impact. In all four subgroups, the attribution for distance to the nearest school is not significantly different from $0$. However, distance to a station does receive significantly more attribution amongst houses that are \textit{near} to the CBD, compared to those farther away.
    \item \label{en:re4} Interestingly, while not a significant result, the number of images used in advertising did appear to receive greater attribution amongst house sales that were \textit{far} from the CBD, compared to those \textit{near} to it.
    \item \label{en:re5} Amongst sales that were \textit{far} from the CBD, land size and roominess both received significantly more attribution than distance to the CBD, in 2019. However, this difference vanished in 2020, with distances to the CBD apparently gaining more attribution, while roominess and land size apparently lost some attribution, in a relative leveling out of these three covariates.
\end{enumerate}

\Cref{en:re1} is perhaps unsurprising: distances were less relevant far from the city, where price variability was influenced more by roominess and land size. Indeed, we can assume we are less likely to find small and expensive houses far from the CBD. However, the authors find \Cref{en:re5} interesting: near the city, the behaviour didn't change significantly during the 2020 period. However, far from the city, the behaviour did change significantly, moving toward the near-city behaviour. Distance to the city became more important for explaining variability in price, while land size and roominess both became less important, compared with their 2019 estimates. Our initial guess at an explanation was that near-city buyers, with near-city preferences, were temporarily motivated by urgency to buy farther out. However, according to \Cref{table:subgroups}, the observed ratio of near-city buyers to non-near buyers actually increased in this period, from $1.26$ in 2019 to $1.61$ in 2020. We will not take this expository analysis further here, but we hope that the interested reader is motivated to take it further, and to this end we have made the data and R, python and Julia scripts available at \url{github.com/BSMLcode/shapley\_confidence}.
\begin{figure}[tb]
    \centering
    \includegraphics[width=\textwidth]{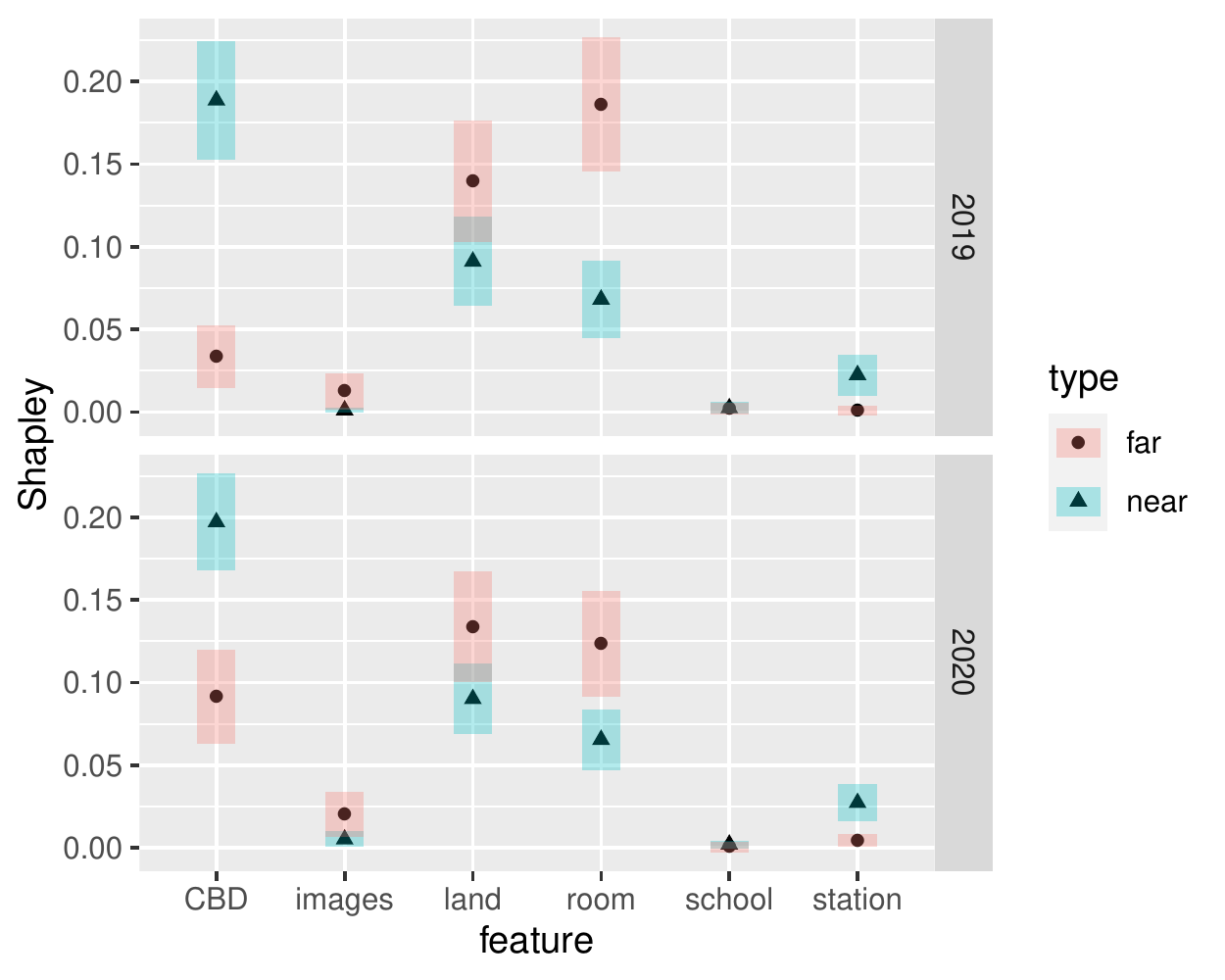}
    \caption{Shapley values and associated asymptotic 95\% confidence intervals for each of the 6 covariates, within the 4 real estate data subgroups (near/far and 2019/2020). Blue bands and triangle markers represent the \textit{near} subgroup (i.e., $\leq 25$km from Central Business District) and red bands with circular markers represent the \textit{far} subgroup ($>25$km).}
    \label{fig:realestate_shapley}
\end{figure}
\begin{table}[H]
    \centering
    \caption{Shapley values and asymptotic CIs for the covariates of the real estate data.}
    \begin{tabular}{lllll}
        \toprule
         Feature & \textit{near} (2019)        & \textit{far} (2019)         & \textit{near} (2020)        & \textit{far} (2020) \\
         \midrule
         CBD     & 0.19 (\phantom{-}0.15,0.23) & 0.03 (\phantom{-}0.02,0.05) & 0.20 (\phantom{-}0.17,0.23) & 0.10 (\phantom{-}0.06,0.12) \\
         images  & 0.00 (-0.00,0.00) & 0.01 ( 0.00,0.02) & 0.06 ( 0.00,0.01) & 0.02 ( 0.01,0.03) \\
         land    & 0.09 ( 0.06,0.12) & 0.14 ( 0.10,0.18) & 0.09 ( 0.07,0.11) & 0.13 ( 0.10,0.17) \\
         school  & 0.00 (-0.00,0.00) & 0.00 (-0.01,0.01) & 0.00 (-0.00,0.00) & 0.00 (-0.00,0.00) \\
         station & 0.02 ( 0.01,0.04) & 0.00 (-0.00,0.00) & 0.03 ( 0.02,0.04) & 0.00 ( 0.00,0.01) \\
         room    & 0.07 ( 0.05,0.09) & 0.19 ( 0.15,0.23) & 0.07 ( 0.05,0.08) & 0.12 ( 0.09,0.16) \\ 
         \bottomrule
    \end{tabular}
    \label{table:realestate_results}
\end{table}{}
\FloatBarrier
% ---------------------------------------------------------
\section{Discussion}\label{sec:discussion}
% ---------------------------------------------------------
In \Cref{sec:main}, we showed that under an elliptical (or pseudo-elliptical) joint distribution assumption, the game theoretic Shapley value decomposition of $R^2(\mathcal{Z}_n)$ is asymptotically normal. Implementing this result, we produced asymptotic Shapley value CIs and hypothesis tests.

In \Cref{s: monte}, we examined the coverage and width statistics of these asymptotic CIs over a range of sample sizes, using Monte Carlo simulations. These simulations were conducted across three separate data generating processes: using a variety of correlations with a compound symmetry covariance matrix under multivariate normal (i) and Student-$t$ (ii) distributions. The simulations were also conducted under a normal distribution data generating process, with random Wishart covariance matrix (iii). In all three cases, the coverage and width statistics were compared to the corresponding statistics for the non-parametric bootstrap CIs. The computation time and memory usage were also benchmarked against the bootstrap CIs. In Section \ref{ss: recommendations} we provided recommendations for when asymptotic CIs should be preferred and used, over the bootstrap CIs. We found that the asymptotic CIs have estimated coverage probabilities of at least $0.85$ across all studies, are preferable over the bootstrap CIs for small sample sizes ($n \leq 50$), and are often (although not always) favourable for large sample sizes. The asymptotic CIs are also far more computationally efficient than bootstrap CIs (at least for the cases of three and five  explanatory variables), and show improved coverage and width when correlation is further from $c = 0.5$.

Finally, in \Cref{sec:realestate}, we demonstrated the application of our derived asymptotic CIs to a data set consisting in house prices from Melbourne, to investigate a period of altered consumer behaviour during the initial stages of the arrival of COVID-19 in Australia. Using the CIs, we identified significant changes in model behaviour between 2019 and 2020, and attributed these changes amongst features, highlighting a potentially interesting direction of future research into the period. We also observed significant changes in behaviour between houses classified as close to the city, and those far from the city. 

We have made the computational implementations of our methods openly available for use. These computational resources are implemented in the R and Julia programming languages. 

We are preparing R, Julia, and Python versions of our methods for release. Our code and future progress regarding these implementations will be made available at \url{github.com/BSMLcode/shapley\_confidence}. Additionally, we aim to use what we have developed in order to derive the asymptotic distributions of variance inflation factors and their generalizations \cite{fox1992generalized}, as well as the closely related Owen values decomposition of the coefficient of determination \cite{Huettner:2012aa}.
\FloatBarrier
\appendix
% ---------------------------------------------------------
\section{Additional figures}\label{appendix:figures}
% ---------------------------------------------------------
%\begin{figure}[H]
%    \centering
%    \includegraphics[width=\textwidth]{Figures/realestate_shapley2.pdf}
%    \caption{Shapley values and associated asymptotic 95\% confidence intervals for each of the 6 covariates, within the 4 %real estate data subgroups (near/far and 2019/2020). Blue bands and triangle markers represent the year 2019 while red %bands with circular markers represent the year 2020.}
%    \label{fig:realestate_shapley2}
%\end{figure}
\begin{figure}[H]
    \centering
    \includegraphics[width=\textwidth]{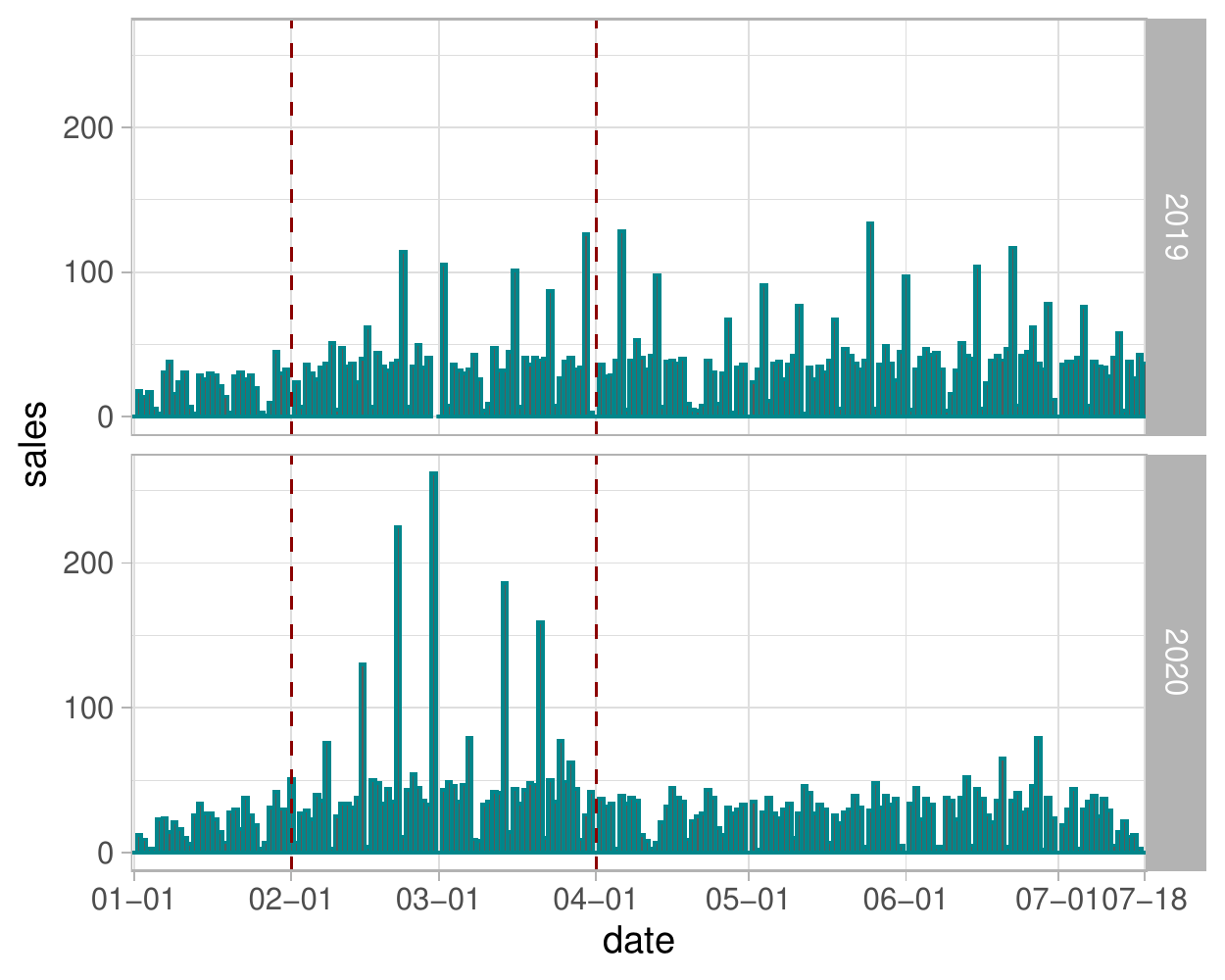}
    \caption{Bar plot of sales per day between 1 January and 18 July in 2019 and 2020. Vertical dashed red lines indicate 1 February and 1 April, between which a spike in sales is observed.}
    \label{fig:dates_barplot}
\end{figure}
\begin{figure}[H]
    \centering
    \includegraphics[width=\textwidth]{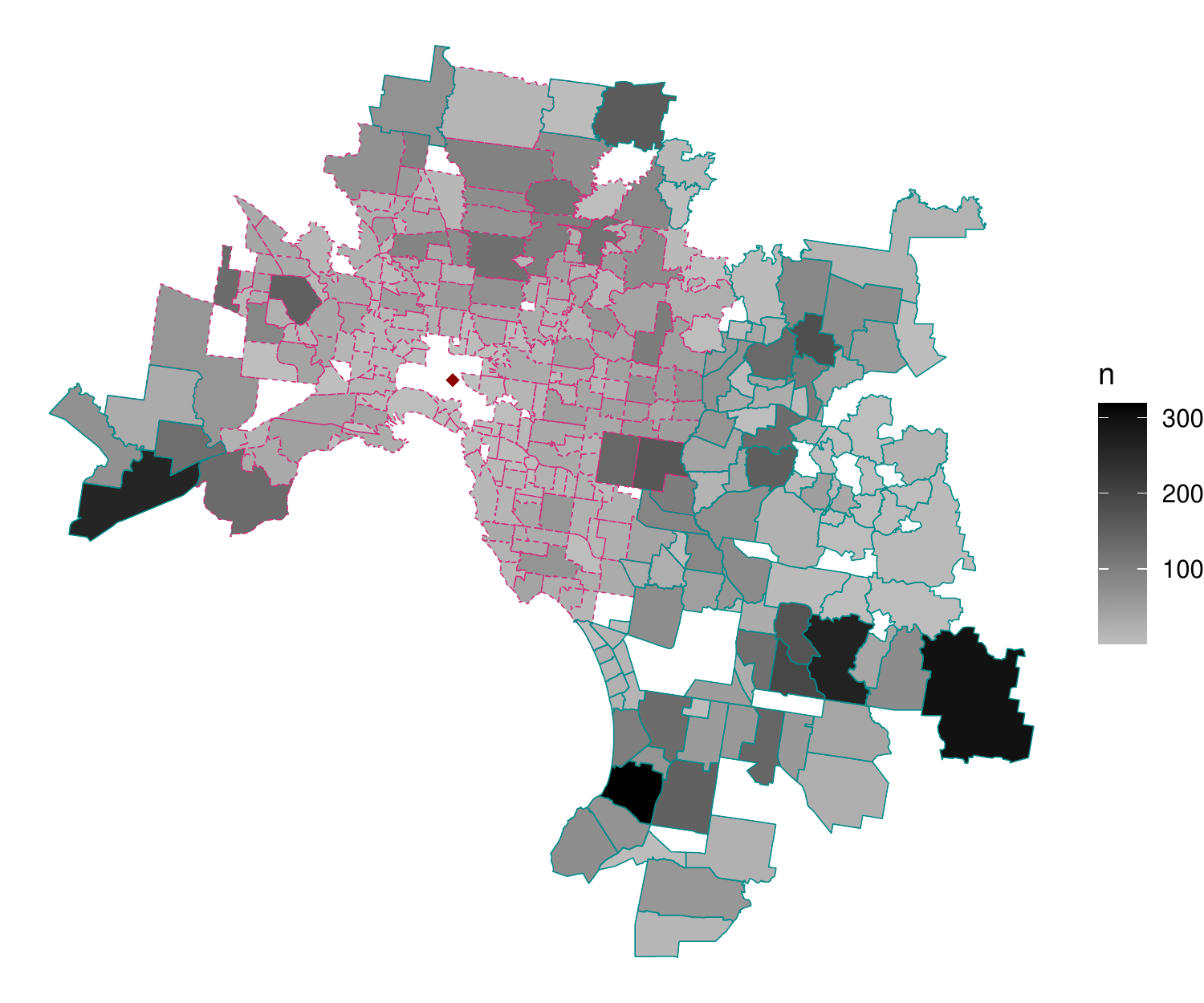}
    \caption{Map of Melbourne suburbs included in this study, with greyscale shading to represent sample sizes. Pink dashed polygon borders contain suburbs classified as \textit{near} to the Central Business District (i.e., $\leq 25$km), and blue solid polygon borders represent those classified as \textit{far} ($> 25$km). The solid red diamond represents the location of the CBD.}
    \label{fig:suburb_map}
\end{figure}
\begin{figure}[H]
    \begin{center}
        \includegraphics[width = \textwidth]{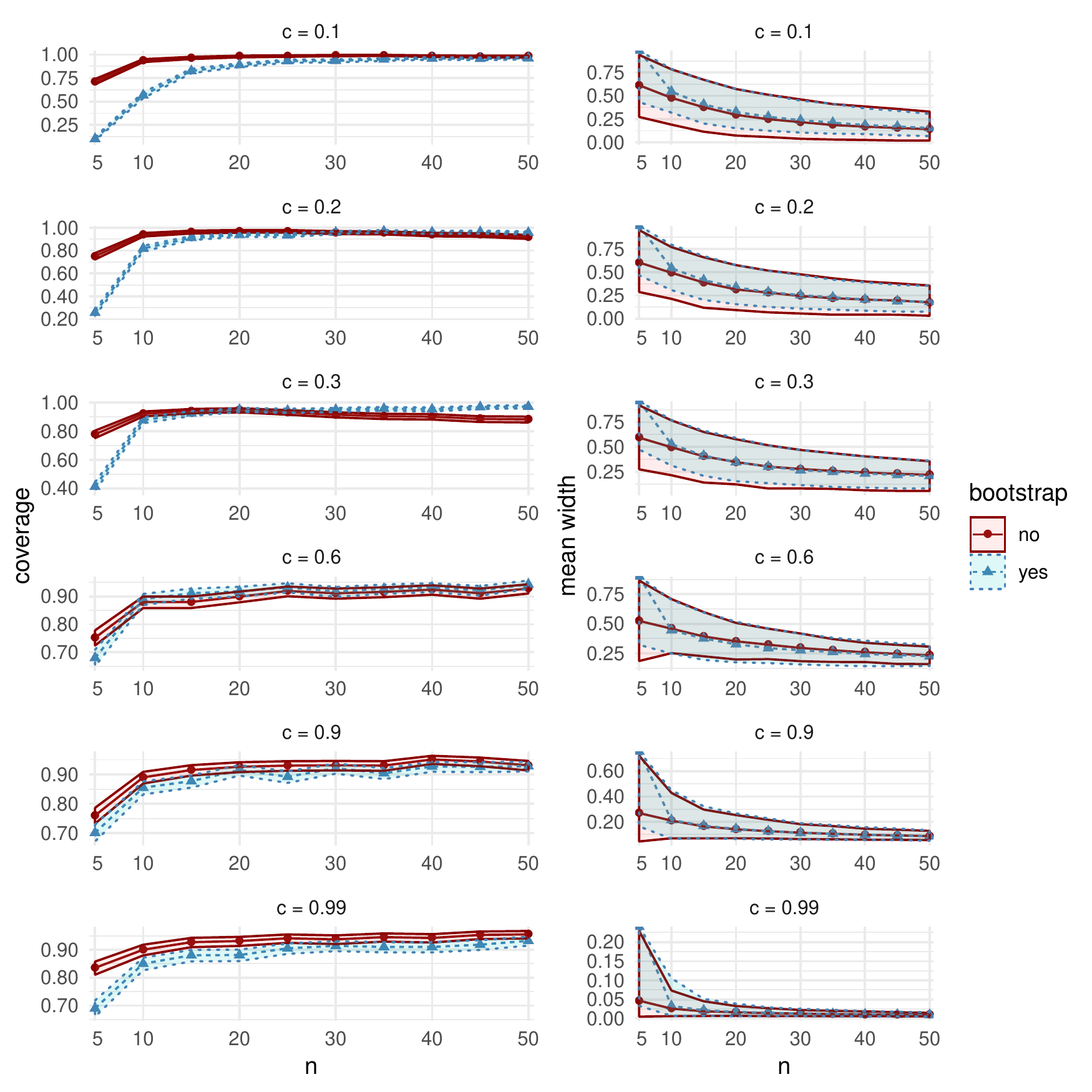}
    \end{center}
    \caption{Comparisons of coverage (left column) and mean width (right column), between the bootstrap CIs (dashed lines with triangular markers) and the asymptotic CIs (solid lines with circular markers) in MC Study A. Rows represent  correlations $c$, increasing from $0.1$ on the top row to $0.99$ on the bottom row. The horizontal axes display the sample sizes $n = 5,10,\ldots,50$.}
    \label{fig: A-small}
\end{figure}
\begin{figure}[H]
    \begin{center}
        \includegraphics[width = \textwidth]{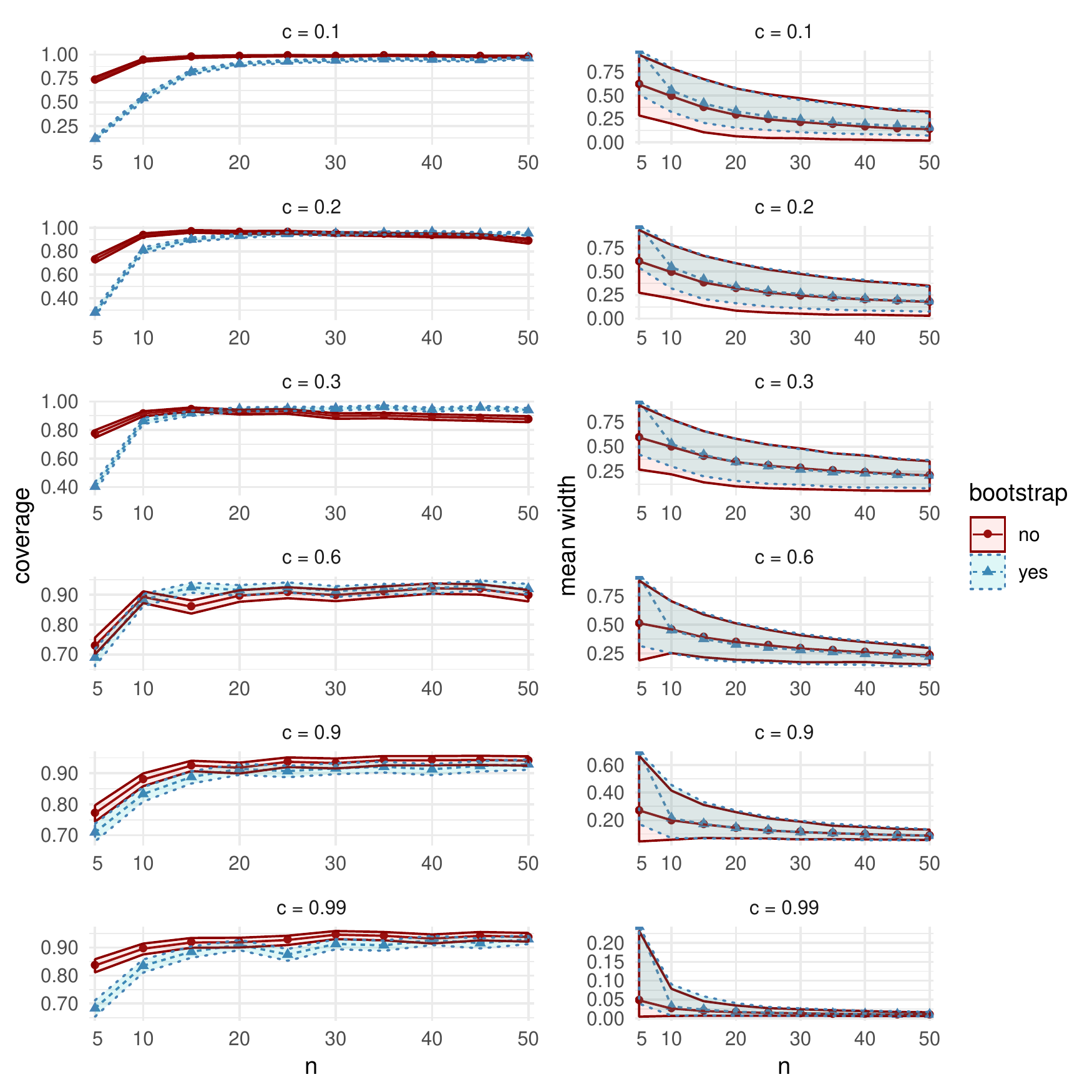}
    \end{center}
    \caption{Comparisons of coverage (left column) and mean width (right column), between the bootstrap CIs (dashed lines with triangular markers) and the asymptotic CIs (solid lines with circular markers) in MC Study B. Rows represent  correlations $c$, increasing from $0.1$ on the top row to $0.99$ on the bottom row. The horizontal axes display the sample sizes $n = 5,10,\ldots,50$.}
    \label{fig: B-small}
\end{figure} 
\begin{figure}[H]
    \begin{center}
        \includegraphics[width = \textwidth]{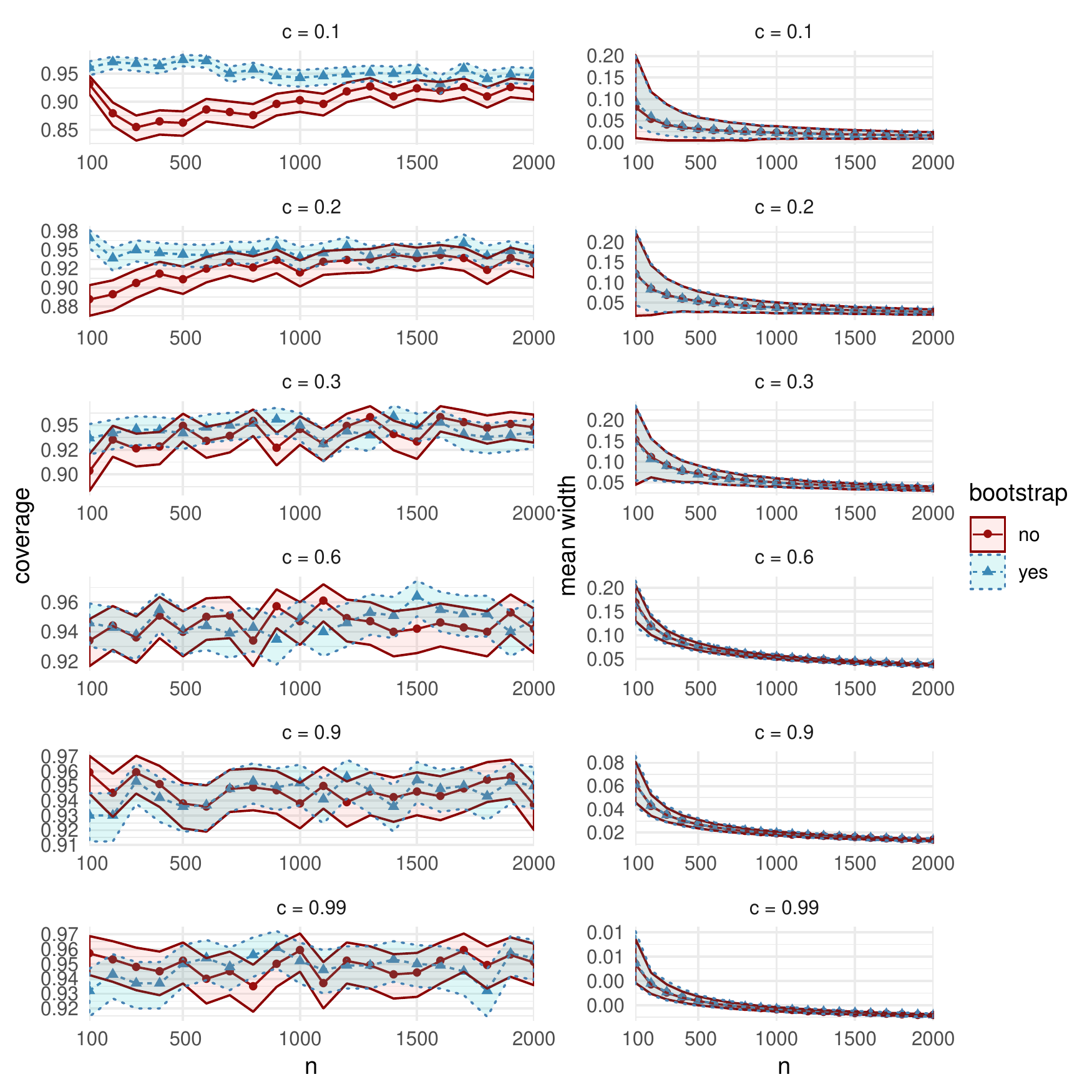}
    \end{center}
    \caption{Comparisons of coverage (left column) and mean width (right column), between the bootstrap CIs (dashed lines with triangular markers) and the asymptotic CIs (solid lines with circular markers) in MC Study B. Rows represent correlations $c$, increasing from $0.1$ on the top row to $0.99$ on the bottom row. The horizontal axes display the sample sizes $n = 100,200,\ldots,2000$.}
    \label{fig: B-large}
\end{figure} 
\begin{figure}[H]
    \begin{center}
        \includegraphics[width = \textwidth]{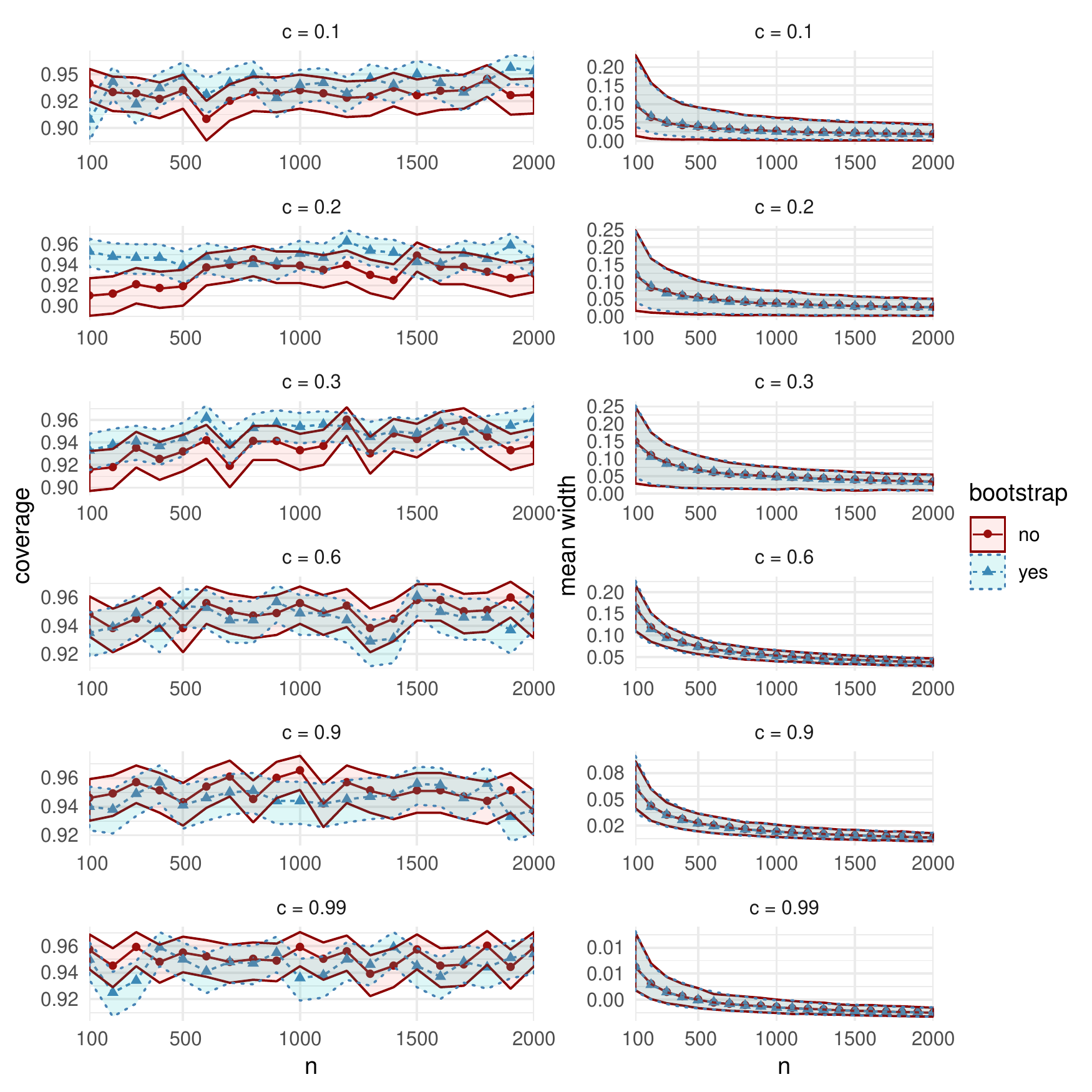}
    \end{center}
    \caption{Comparisons of coverage (left column) and mean width (right column), between the bootstrap CIs (dashed lines with triangular markers) and the asymptotic CIs (solid lines with circular markers) in MC Study C. Rows represent correlations $c$, increasing from $0.1$ on the top row to $0.99$ on the bottom row. The horizontal axes display the sample sizes $n = 100,200,\ldots,2000$.}
    \label{fig: C-large}
\end{figure} 
\bibliographystyle{apalike} %{plain}
\bibliography{20191030_bibliography}

\end{document}